%% file: main.tex
\documentclass[10pt,twocolumn,letterpaper]{article}

 \usepackage[pagenumbers]{cvpr} 

\input{preamble}

\definecolor{cvprblue}{rgb}{0.21,0.49,0.74}
\usepackage[pagebackref,breaklinks,colorlinks,citecolor=cvprblue]{hyperref}

\usepackage{algpseudocode}
\usepackage{algorithm}
\usepackage{algorithmicx}
\usepackage{amsthm}

\makeatletter
\def\thanks#1{\protected@xdef\@thanks{\@thanks
        \protect\footnotetext{#1}}}
\makeatother

\algblock[Block]{Block}{EndBlock}
\algblockdefx[Block]{Block}{EndBlock}%
[1]{{#1}}%
[1]{{#1}}

\newcommand{\D}{\operatorname{D}}
\newcommand{\Dg}{\operatorname{D}_{\mathcal{G}}}

\newcommand{\G}{\mathcal{G}}

\newtheorem{theorem}{Theorem}

\newtheorem{proposition}[theorem]{Proposition}

\usepackage[capitalize]{cleveref}
\crefname{section}{Sec.}{Secs.}
\Crefname{section}{Section}{Sections}
\Crefname{table}{Table}{Tables}
\crefname{table}{Tab.}{Tabs.}
\crefname{equation}{}{}

\title{Equivariant plug-and-play image reconstruction}

\author{Matthieu Terris, Thomas Moreau,
\\ Université Paris-Saclay, Inria, CEA
\\ Palaiseau, 91120, France
\and  
Nelly Pustelnik, Julian Tachella\\
ENSL, CNRS UMR 5672\\
Lyon, 69342, France}

\begin{document}
\maketitle

\begin{abstract}
Plug-and-play algorithms constitute a popular framework for solving inverse imaging problems that rely on the implicit definition of an image prior via a denoiser.  
These algorithms can leverage powerful pre-trained denoisers to solve a wide range of imaging tasks, circumventing the necessity to train models on a per-task basis. 
Unfortunately, plug-and-play methods often show unstable behaviors, hampering their promise of versatility and leading to suboptimal quality of reconstructed images.
In this work, we show that enforcing equivariance to certain groups of transformations (rotations, reflections, and/or translations) on the denoiser strongly improves the stability of the algorithm as well as its reconstruction quality. We provide a theoretical analysis that illustrates the role of equivariance on better performance and stability. 
We present a simple algorithm that enforces equivariance on any existing denoiser by simply applying a random transformation to the input of the denoiser and the inverse transformation to the output at each iteration of the algorithm.
Experiments on multiple imaging modalities and denoising networks show that the equivariant plug-and-play algorithm improves both the reconstruction performance and the stability compared to their non-equivariant counterparts.
\end{abstract}

\section{Introduction}

Linear inverse imaging problems are ubiquitous in imaging sciences, famous instances of which include image restoration, magnetic resonance imaging (MRI), computed tomography, and astronomical imaging to name a few.
In this setting, the aim is to recover an image $x\in\mathbb{R}^n$ from measurements $y\in\mathbb{R}^m$ acquired through
\begin{equation}
    y = Ax+\epsilon,
\label{eq:inv_pb}
\end{equation}
where $A\colon \mathbb{R}^n\to \mathbb{R}^m$ is a linear operator and $\epsilon$ the realisation of some random noise.
A myriad of methods for solving problems of the likes of \eqref{eq:inv_pb} have been proposed in the literature, ranging from variational methods \cite{combettes2014forward, chambolle2016introduction} (solution of a cost function minimisation problem) to end-to-end reconstruction with deep neural networks \cite{jin2017deep, knoll2020fastmri, zhang2023ntire} and, more recently, diffusion algorithms \cite{chung2022diffusion, kawar2022denoising, zhu2023denoising, de2021diffusion}. 

In this work, we focus on approaches relying on implicit denoising priors. For instance, plug-and-play (PnP) algorithms
\cite{venkatakrishnan2013plug, ryu2019plug, zhang2021plug, kamilov2023plug} propose to replace the proximity operator involved in variational methods (solved by proximal algorithms) by a denoiser modeling an implicit image prior. Similarly, regularization-by-denoising (RED) approaches propose to replace gradient steps on the prior term by a denoiser \cite{romano2017little}.
While the denoiser is typically trained as a Gaussian denoiser using grayscale or color natural images, it can be plugged into algorithms designed to tackle a wide range of image-related problems, without being constrained by the nature of the input images (e.g. MRI images, CT scans, etc) \cite{ahmad2020plug, liu2022recovery, terris2023image, stergiopoulou2023fluctuation}.  

\begin{figure*}[t]
    \centering
    \includegraphics[width=0.99\linewidth]{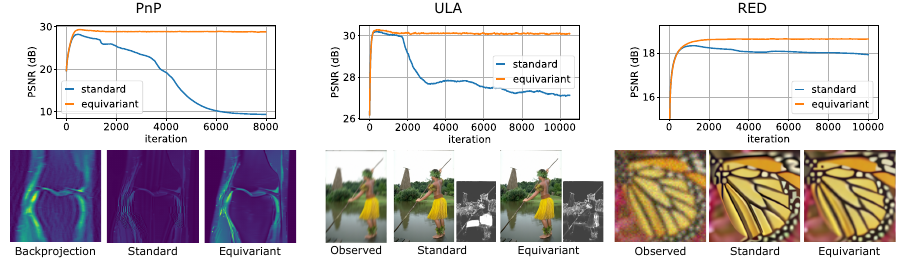}
    \vspace{-1em}
\caption{\textbf{Instability of algorithms relying on implicit denoising priors can be solved by incorporating equivariance.}
Enforcing approximate equivariance of the denoiser at test time allows to both stabilize the algorithm and to improve the reconstruction quality without needing to retrain the implicit prior. \textbf{Left:} PnP algorithm applied to an accelerated MRI problem. \textbf{Middle:} Unadjusted Langevin sampling algorithm for a motion blur problem; estimated mean and variance of the associated Markov chain are displayed. \textbf{Right:} RED algorithm on a $4\times$ super-resolution problem.}
\vspace{-1em}
\label{fig:descriptive_fig}
\end{figure*}

Yet, despite successful performance in various applications, these algorithms tend to suffer from instability issues, making it difficult to transfer across imaging tasks or to derive statistical estimates from the reconstruction. In particular, RED and PnP algorithms often require careful finetuning \cite{rick2017one, meinhardt2017learning}, or departing from the original optimization algorithm for efficient application \cite{zhang2021plug, cohen2021regularization}.
Thus, several works have been proposed to restore the convergence of PnP and RED algorithms while establishing a clear connection with an associated cost function \cite{ryu2019plug, terris2020building, cohen2021regularization, pesquet2021learning, hurault2021gradient}.
However, these efforts introduce notable constraints on the denoiser, leading to a trade-off between stability and reconstruction performance. Typically, convergent PnP and RED algorithms tend to transfer easily to 
new problems, but perform less well than their fine-tuned and early-stopped non-convergent counterparts. 

In this work, we propose to investigate theoretically and empirically the effect of equivariance on algorithms relying on implicit denoising priors.
More precisely, we prove that enforcing equivariance on the denoiser improves the stability of the resulting PnP algorithm and plays a symmetrizing role on the Jacobian of the implicit prior.
Our experiments show that the proposed equivariant approach can significantly improve the quality of images reconstructed with PnP algorithms as well as with popular algorithmic frameworks such as RED or Langevin algorithms.
We give an overview of the possibilities offered by the proposed approach on popular algorithms relying on implicit denoising priors in Figure~\ref{fig:descriptive_fig}.

\section{Related Work}

\paragraph{Stable plug-and-play algorithms}
A significant amount of work has recently been proposed to provide stable PnP algorithms. A popular line of research in that direction consists in constraining the denoiser's architecture to ensure its stability when plugged into a PnP scheme. 
Some approaches propose to regularize the denoiser's training loss with a term penalizing the Lipschitz constant of the denoiser~\cite{ryu2019plug, pesquet2021learning}, which can be combined with architectural constraints \cite{terris2020building, hurault2021gradient}. Other works propose to alter the optimization algorithm itself in order to ensure its stability \cite{meinhardt2017learning, cohen2021regularization, zhang2021plug, hurault2023convergent}. 
In turn, ensuring convergence of the algorithm allows to ensure better transferability to new imaging tasks \cite{terris2023image}, 
but also to perform iteration-intensive tasks, such as sampling
from the posterior distribution \cite{laumont2022bayesian}.
However, all the aforementionned methods come at the cost of either algorithmic modifications or strong constraints on the design of the denoiser, as opposed to the original, denoiser agnostic approaches \cite{venkatakrishnan2013plug, romano2017little}.

\paragraph{Equivariance in imaging inverse problems}
Equivariance to certain transformations, such as rotations or translations, has often been a desired property when designing handcrafted variational priors \cite{selesnick2005dual, condat2017discrete, saydjari2022equivariant}. 
Geometric ensembling techniques have been known in the computer vision literature, where it has been shown to improve reconstruction quality for image super-resolution \cite{timofte2016seven}. 
In the case of linear denoisers, symmetrization of the Jacobian has been shown to improve performance of linear denoisers \cite{milanfar2013symmetrizing}.
Yet, to the notable exception of~\cite{zhang2021plug, terris2023image}, this strategy went unnoticed in the PnP literature, and to the best of our knowledge, the role of equivariance has not been explored in the context of algorithms relying on implicit denoising priors.
In a different line of work, but still in the context of imaging inverse problems, recent works have exploited equivariance for the design of unrolled network architectures~\cite{celledoni2021equivariant}, or the construction of self-supervised learning losses~\cite{chen2021equivariant} (e.g., see the recent review~\cite{chen2023imaging}). 

\paragraph{Equivariant neural networks}
More generally, invariance of the underlying prior appears as a natural assumption in a large number of applications, and there exists a vast literature on building equivariant neural network architectures~\cite{cohen2016group,weiler2019general,bronstein2021geometric}. 
Typical applications involve segmentation on spherical manifolds, robotics, point cloud analysis, data augmentation to name a few \cite{fuchs2020se, chen2021equivariant, muller2021rotation, rommel2022deep, wang2022equivariant, zhu20234d}.
However, these networks often fail to perform as well as other state-of-the-art architectures. Moreover, it can be challenging to incorporate complex layers (e.g., upsampling/downsampling, attention-like layers, etc.) without breaking the equivariance of the resulting network. In this work, we provide a simple method for rendering any denoiser equivariant, without any architectural constraints.

\section{PnP algorithms}

Traditional variational approaches for solving \eqref{eq:inv_pb} consist in reformulating it as a minimization problem. Following a maximum-a-posteriori approach, one can derive an estimate $\widehat{x}$ as
\begin{equation}
    \widehat{x} = \underset{x}{\operatorname{argmin}}\,\, f(x)+ \lambda r(x)
\label{eq:min_pb}
\end{equation}
where $f$ is a data-fidelity enforcing term, $r$ is a regularization enforcing prior knowledge about the solution, and $\lambda>0$ is a regularization parameter.

PnP approaches propose to replace the proximity operator of $r$ \cite{bauschke2017correction} (implicit gradient step)
arising in algorithms for solving \eqref{eq:min_pb}
by a denoiser $\operatorname{D}$~\cite{venkatakrishnan2013plug}. For the standard case of a quadratic data-fidelity term $f(x) = \frac{1}{2}\|Ax-y\|_2^2$, the classical PnP algorithm reads
\begin{equation}
x_{k+1} = \D\left(x_k-\gamma A^{\top} (A x_k-y)\right) \tag{PnP}
\label{eq:pnp_fb}
\end{equation}
where $\gamma >0$ is a stepsize. 

Similarly, regularization by denoising (RED) algorithms approximate the gradient of $r$ as $\nabla r(x)\propto x-\D(x)$ using Tweedie's formula~\cite{romano2017little}. A simple explicit gradient descent optimization based on this definition yields
\begin{equation}
x_{k+1} = x_k - \gamma A^{\top} (A x_k-y)- \gamma \lambda (x_k-\operatorname{D}(x_k)). \tag{RED}
\label{eq:red_gd}
\end{equation}

We also consider the Unadjusted Langevin algorithm (ULA)~\cite{laumont2022bayesian}, which aims to obtain samples associated with the negative log posterior density $-\log p (x|y) \propto f(x) + \lambda r(x)$, and requires adding noise to the iterates in \cref{eq:red_gd}, i.e.,
\begin{multline}
x_{k+1} = x_k - \gamma A^{\top} (A x_k-y)- \gamma \lambda (x_k-\operatorname{D}(x_k)) \\
+ \sqrt{2\gamma} \epsilon_k, \tag{ULA}
\label{eq:ula}
\end{multline}
where $\epsilon_k \sim \mathcal{N}(0,I)$ is a standard Gaussian vector.

While there exist many different variants of \Cref{eq:pnp_fb}, \Cref{eq:red_gd}, and \Cref{eq:ula}, we here focus on their most standard formulations. Interestingly, these algorithms have shown impressive performance on a wide variety of imaging tasks while relying on Gaussian denoisers agnostic to the imaging modality of interest. 
Yet, these algorithms suffer from a lack of stability and potential divergence effects, hurting their versatility. 

\subsection{Proposed equivariant approach}

Intuitively, imaging priors should have some invariance properties with respect to certain groups of transformations, such as rotations, translations, and reflections. We denote transformations associated with a group $\G$ as  $\{T_g\}_{g\in \G}$ where $T_g \in \mathbb{R}^{n\times n}$ is a unitary matrix\footnote{While it is possible to define group actions with non-unitary matrices~\cite{serre1977linear}, here we focus on the unitary matrices, which is the case of translations, rotations and reflections of images.}. We say that $\D$ is equivariant to the group action $\{T_g\}_{g\in\G}$ if $\D(T_gx) = T_g\D(x)$ for all $x$ and $g\in\G$.  At the algorithmic level, this requirement translates into the equivariance of the denoiser with respect to the transforms of interest: if $r(x)$ is a $\G$-invariant function, its proximal operator and gradient (if they exist) are necessarily $\G$-equivariant functions \cite{celledoni2021equivariant}. 

A simple way of rendering any function $\G$-equivariant is by averaging over the group\footnote{This construction is known as Reynolds averaging, see e.g. \cite{serre1977linear}.}. The associated averaged denoiser 
\begin{equation} \label{eq: averaged den}
    \Dg(x) \overset{\text{def}}{=} \frac{1}{|\G|} \sum_{g\in \G} T_g^{-1}\D(T_gx).
\end{equation}
is equivariant by construction. For large groups, or in the case of large denoising architectures $\D$, computing the averaged denoiser might be too computationally demanding. However, in this work, we propose to use a simple Monte Carlo approximation by sampling a single transformation at each step of the algorithm, i.e., 
\begin{equation}
\begin{aligned}
g &\sim \mathcal{G} \\
\widetilde{\operatorname{D}}_{\G}(x) &= T_g^{-1}\operatorname{D}(T_g x).
\end{aligned}
\label{eq:mc_denoise}
\end{equation}

The equivariant counterpart of the \cref{eq:pnp_fb}, \cref{eq:red_gd} and \cref{eq:ula} algorithms is thus simply obtained by replacing the denoiser $\D$ by a sample of the Monte Carlo estimate \eqref{eq:mc_denoise} at each iterate of the algorithm. Explicit versions of these algorithms can be found in the Supplementary Material (SM).

\section{Theoretical analysis}
In this section, we provide a theoretical analysis of the advantages in terms of performance and stability of equivariant denoisers compared to their non-equivariant counterparts. In this section, we denote the Jacobian of the denoiser as $J_x \overset{\text{def}}{=}  \frac{\delta \operatorname{D}}{\delta x}(x)$.

\paragraph{Optimality of equivariant denoisers}
We first show that if the signal distribution is $\G$-invariant, then for any denoiser $\D$, its averaged version \cref{eq: averaged den} obtains an equal or better denoising performance. This can be shown by computing the expected $\ell_2$ error with respect to the signal and noise distributions, i.e.,
\begin{align*}
&\mathbb{E} \, \Big\| \frac{1}{|\G|}\sum_{g\in\G} T_g^{-1}\D (T_g (x+\varepsilon)) - x \Big\| \\
&\leq  \frac{1}{|\G|}  \sum_{g\in\G} \mathbb{E}\, \| T_g^{-1}\D(T_g (x+\varepsilon) - x \|  \\
&\leq  \frac{1}{|\G|}  \sum_{g\in\G} \mathbb{E}\, \| T_g^{-1}\D(T_g T_g^{-1}(x+\varepsilon)) - T_g^{-1}x \|  \\
&\leq \mathbb{E}\, \| \D(x+\varepsilon) - x \| 
\end{align*}
where the expectation is taken with respect to $x$ and $\varepsilon$, and where
the second line uses the triangle inequality and the third line uses that (i) $\mathbb{E} h(x) =\mathbb{E} h(T_g^{-1}x)$ for any function $h:\mathbb{R}^{n}\mapsto \mathbb{R}$ and $T_g$ in the group action if the distribution of $x$ and $\varepsilon$ is $\G$-invariant, and that (ii) the transformations $T_g$ are isometries.

\paragraph{Existence of an explicit prior}

A necessary condition for a denoiser to be associated with an explicit (PnP or RED) prior is to have a symmetric Jacobian, i.e., $J_x = J_x^{\top}$, see~\cite[Theorem 1]{reehorst2018regularization}. Unfortunately, most state-of-the-art denoisers do not exhibit this property. Averaging a denoiser over a sufficiently large group can lead to symmetric Jacobians. In particular, if the denoiser is linear and $\G$ includes translations and reflections, then the denoiser is assured to have a symmetric Jacobian:

\begin{proposition}
Any linear denoiser $\D$ that is equivariant to the action of 2-dimensional shifts, and vertical and horizontal reflections, has a symmetric Jacobian.
\end{proposition}

\begin{proof}
    Let $\D_{\G}(x)=Mx$ with Jacobian $M\in\mathbb{R}^{n\times n}$. A matrix $M$ that is equivariant to the action of 2-dimensional has a circulant form, i.e., $M=\text{circ}(d)$ where $d\in\mathbb{R}^n$ is a filter. Thus the transposed Jacobian is also a circulant matrix  $M^{\top} = \text{circ}(d')$ where $d'\in \mathbb{R}^n$ is the transposed filter. Since $D$ is also equivariant to vertical and horizontal reflections, we have that $d$ is even and real, and thus $M=M^{\top}$. 
\end{proof}

While this result applies only to linear denoisers, the symmetry of non-linear denoisers also improves when incorporating equivariance: \Cref{tab: symmetry jac} shows the relative symmetry error $\| J_x - J_x^{\top}\|_{\mathrm{F}}^2/\|J_x\|_{\mathrm{F}}^2$ of popular non-linear denoisers averaged over 10 different patches of $64 \times 64$ pixels. The $\G$-equivariant denoisers $\Dg$ have significantly smaller errors than their non-equivariant counterparts.

\begin{table}[h]
\centering
\footnotesize
\begin{tabular}{@{\hskip 1pt}l @{\hskip 15pt} c @{\hskip 7pt} c @{\hskip 7pt} c @{\hskip 7pt} c @{\hskip 7pt} c}
\toprule
 & \begin{tabular}[c]{@{}c@{}}Lipschitz\\ DnCNN\end{tabular} & DnCNN & SCUNet & SwinIR & \begin{tabular}[c]{@{}c@{}}DRUNet\\ ($\sigma_d$=0.01)\end{tabular} \\  %
 \hline
Standard $\operatorname{D}$ & 0.014 & 0.022 & 0.954 & 0.604 & 0.030 \\ 
Equivariant $\Dg$ & 0.003 & 0.005 & 0.710 & 0.291 & 0.008 \\ 
\bottomrule
\end{tabular}
\vspace{-0.5em}
\caption{Mean Jacobian Symmetry error $\|J_x-J_x^{\top} \|_F^2/\|J_x\|_F^2$. Equivariant denoisers are obtained by averaging over the group of 90-degree rotations and reflections.}
\vspace{-2em}
\label{tab: symmetry jac}
\end{table}

\paragraph{Lipschitz constant of the denoiser}
The stability of PnP algorithms depends crucially on the Lipschitz constant of the denoiser~\cite{ryu2019plug, kamilov2023plug}. For example, if the Lipschitz constant of the denoiser is lower than 1, both the \cref{eq:pnp_fb} and \Cref{eq:red_gd} iterates converge under a good choice of step size. 
Since the Lipschitz constant of the sum of two mappings is smaller or equal than the sum of their Lipschitz constants,
we have that the Lipschitz constant of the averaged equivariant denoiser is necessarily equal or lower than the non-equivariant one.  If we restrict to linear denoisers, we can show that the equivariant denoiser will have a strictly smaller constant, as long as the dominant singular vector is not equivariant:

\begin{proposition}
     Let $\D(x)=Mx$ be a linear denoiser with singular value decomposition $M=\sum_{i=1}^{n} \lambda_i u_i v_i^{\top}$ and $\lambda_1 > \lambda_2 \geq \dots \geq \lambda_n \geq 0$. If the principal component $u_1v_1^{\top}$ is not $\G$-equivariant, the averaged denoiser $\operatorname{D}_{\G}$ has a strictly smaller Lipschitz constant than $\operatorname{D}$. 
\end{proposition}

\begin{proof}
    For any $p$ matrices $A_1,\dots,A_p$, we have that $\|\frac{1}{p}\sum_{g=1}^{p} A_g\| = \frac{1}{p} \sum_{g=1}^p\|A_g\|$ if and only if all matrices share the same leading left and right singular vectors. The $\G$-averaged denoiser can be written as $\Dg= \frac{1}{|\G|}\sum_{g=1}^{|\G|} A_g$ where $A_g:=T_gMT_g^{-1}$. We have that $A_g$ has the same singular values as $M$ since singular vectors are defined as  $u_i'=T_gu_i$ and $v'=T_gv_i$ for $i=1,\dots,n$. Since $u_1v_1^{\top}$ is not equivariant, we have that $T_gu_1v_1^{\top}T_g^{-1}\neq u_1 v_1$ for some $g\in \G$. Thus, there exist at least 2 terms in the sum $\sum_g T_g M T_g^{-1}$ which do not share the same leading singular vectors, and consequently $\|M\|> \|\frac{1}{|\G|}\sum_g T_g M T_g^{-1} \|$.
\end{proof}

In practice, we observe a significantly smaller constant for most popular non-linear denoisers. \Cref{tab: lip constant} shows the local Lipschitz constant (i.e., the spectral norm of the Jacobian) of various denoisers averaged over 10 different patches of $64 \times 64$ pixels. The averaged denoiser can have a significantly smaller constant.

\begin{table}[h]
\centering
\footnotesize
\begin{tabular}{@{\hskip 1pt}l @{\hskip 15pt} c @{\hskip 7pt} c @{\hskip 7pt} c @{\hskip 7pt} c @{\hskip 7pt} c}
\toprule
 & \begin{tabular}[c]{@{}c@{}}Lipschitz\\ DnCNN\end{tabular} & DnCNN & SCUNet & SwinIR & \begin{tabular}[c]{@{}c@{}}DRUNet\\ ($\sigma_d$=0.01)\end{tabular} \\ \hline %
 Standard $\operatorname{D}$ & 1.06 & 1.44 & 5.78 & 6.28 & 1.57 \\ 
Equivariant $\Dg$ & 0.92 & 1.18 & 4.19 & 4.05 & 1.26 \\ 
\bottomrule
\end{tabular}
\vspace{-0.5em}
\caption{Local Lipschitz constant of the denoiser averaged over 16 image patches. Equivariant denoisers are obtained by averaging over the group of 90 degree rotations and reflections.}
\vspace{-2em}
\label{tab: lip constant}
\end{table}

\paragraph{Interplay between the group action and forward operator}
So far we have focused on the properties associated with an equivariant denoiser, however,  the (lack of) equivariance of $A$ also plays an important role in the convergence of PnP algorithms.
The iterates in \Cref{eq:pnp_fb} can converge even for denoisers with Lipschitz constant larger than 1, as long as the Lipschitz constant of the composition 
\begin{equation}
\label{eq: lip DA}
    \operatorname{D}\circ \, (I-\gamma A^{\top}A)
\end{equation}
is smaller than 1. If the spectra of the matrices $A^{\top}A$ and $\operatorname{D}$ are incoherent, i.e., $A^{\top}A$ and the Jacobian of the denoiser are diagonalized in different bases, the Lipschitz constant of \Cref{eq: lip DA} is likely to be smaller than that of $\operatorname{D}$. A similar stability phenomenon happens for the \Cref{eq:red_gd} and \Cref{eq:ula} iterates due to the incoherence between the spectrum of the forward operator and the one of the equivariant denoiser. 
The following proposition formalizes this intuition in the case of linear denoisers:

\begin{proposition}
\label{prop:composition}
Let $x$ be a grayscale image and $\{T_g\}_{g\in\G}$ be a group of transformations that includes 2-dimensional shifts and $\operatorname{D}_{\G}$ be a linear $\G$-equivariant denoiser. If $A^{\top}A$ is not $\G$-equivariant, it does not share the same singular vectors as the Jacobian of $\operatorname{D}$.
\end{proposition}

\begin{proof}
Let $\D(x)=Mx$ with Jacobian $M\in\mathbb{R}^{n\times n}$.  A matrix $B\in\mathbb{R}^{n\times n}$ is $\G$-equivariant to the action of 2-dimensional shifts if and only if it admits a diagonal decomposition as $B=F \text{diag}(d) F^{*}$ where $F$ is the 2-dimensional Fourier transform, see e.g.~\cite[Section 4.1]{tachella2023sensing}. Thus, if $\operatorname{D}$ is $\G$-equivariant, then $M$ is diagonal in the Fourier domain, whereas the non-equivariant $A^{\top}A$ matrix does not admit such a diagonalization.
\end{proof}

\Cref{tab: lips constant AD} shows the local Lipschitz constant of the mapping in \Cref{eq: lip DA} for various popular non-linear denoisers, where $A$ is a random inpainting operator, which is not equivariant to shifts, rotations nor reflections. The constants are smaller than those shown in~\Cref{tab: lip constant} and are below 1, ensuring contraction of the PnP iterates for that specific operator. The equivariant denoisers show smaller constants than the non-equivariant counterparts.

\begin{table}[h]
\centering
\footnotesize
\begin{tabular}{@{\hskip 1pt}l @{\hskip 15pt} c @{\hskip 7pt} c @{\hskip 7pt} c @{\hskip 7pt} c @{\hskip 7pt} c}
\toprule
 & \begin{tabular}[c]{@{}c@{}}Lipschitz\\ DnCNN\end{tabular} & DnCNN & SCUNet & SwinIR & \begin{tabular}[c]{@{}c@{}}DRUNet\\ ($\sigma_d$=0.01)\end{tabular} \\ \hline %
Standard & 0.91  & 0.91 &  0.62 & 0.69 & 0.83 \\ 
Equivariant   & 0.79  & 0.78 & 0.52 & 0.67  &  0.70 \\ 
\bottomrule
\end{tabular}
\vspace{-0.5em}
\caption{Local Lipschitz constant of PnP iteration $\|J_x(I- A^{\top}A)\|$, with $A$ a random inpainting operator. Equivariant denoisers are obtained by averaging over the group of 90 rotations and reflections.}
\label{tab: lips constant AD}
\vspace{-2em}
\end{table}

\subsection{Non-linear example}
We demonstrate some of the properties analysed in the previous subsection on a non-linear setting of a  neural network denoiser with a single hidden layer.
More precisely, we consider the case where $\operatorname{D}$ is a slight perturbation of a proximity operator, i.e.
\begin{equation}
    \operatorname{D}(x) = B_2 \operatorname{prox}_{\gamma \lambda \|\cdot\|_1}(B_1 x)
\label{eq:approx_prox}
\end{equation}
where $B_1\!\in\!\mathbb{R}^{n\times n}$ 
satisfies $B_1 B_1^\top = B_1^\top B_1 = I$ and where $B_2 = (B_1^\top + P)$ where $P\in\mathbb{R}^{n\times n}$ is a small random perturbation. In particular, if $P=0$, then the denoiser is a well-defined proximal operator, i.e. $\operatorname{D}(x)=\operatorname{prox}_{\gamma \lambda \|B_1 \cdot \|_1}(x)$, with a well-defined prior $r(x)=\|B_1x\|_1$. We stress that, for $P\neq0$, there exists a priori no loss function $g$ associated with the~\cref{eq:pnp_fb} iterates with the denoiser in \eqref{eq:approx_prox}. Furthermore, we assume that $B_1$ is $\G$-equivariant.

Using the fact that $B_1$ and $\operatorname{prox}_{\gamma \lambda \|\cdot\|_1}$ are $\G$-equivariant functions, we can write the $\G$-equivariant denoiser as (see a detailed derivation in the SM)
\begin{align*}
\Dg(x) &= (B_1^{\top}+P_{\G})\operatorname{prox}_{\gamma \lambda \|\cdot\|_1}(B_1 x) 
\end{align*}
where $P_{\G} = \frac{1}{|\G|}\sum_{g\in\G} T_g^{-1}PT_g$ is the $\G$-averaged perturbation. If the original perturbation is not equivariant we have $0 \leq \|P_{\G}\|_F^2 < \|P\|_F^2$, and the equivariant denoiser will be closer to the proximal operator $\operatorname{prox}_{\gamma \lambda \|B_1 \cdot \|_1}(x)$.

Figure~\ref{fig:toy_exp} illustrates the behaviour of the \Cref{eq:pnp_fb} sequence with and without the group averaging in~\cref{eq: averaged den} for two specific choices of $A$ in a 2D toy example, where $\G$ is the group of flips (see the SM for more details). In both cases, the algorithm involving the equivariant denoiser converges to a point close to the global minima associated with the prior $\lambda r(x)=\lambda \|B_1x\|_1$, whereas sequence generated by the non-equivariant algorithm diverges.

\begin{figure}
\small
\begin{minipage}{0.45\textwidth}
    \includegraphics[width=1.0\textwidth]{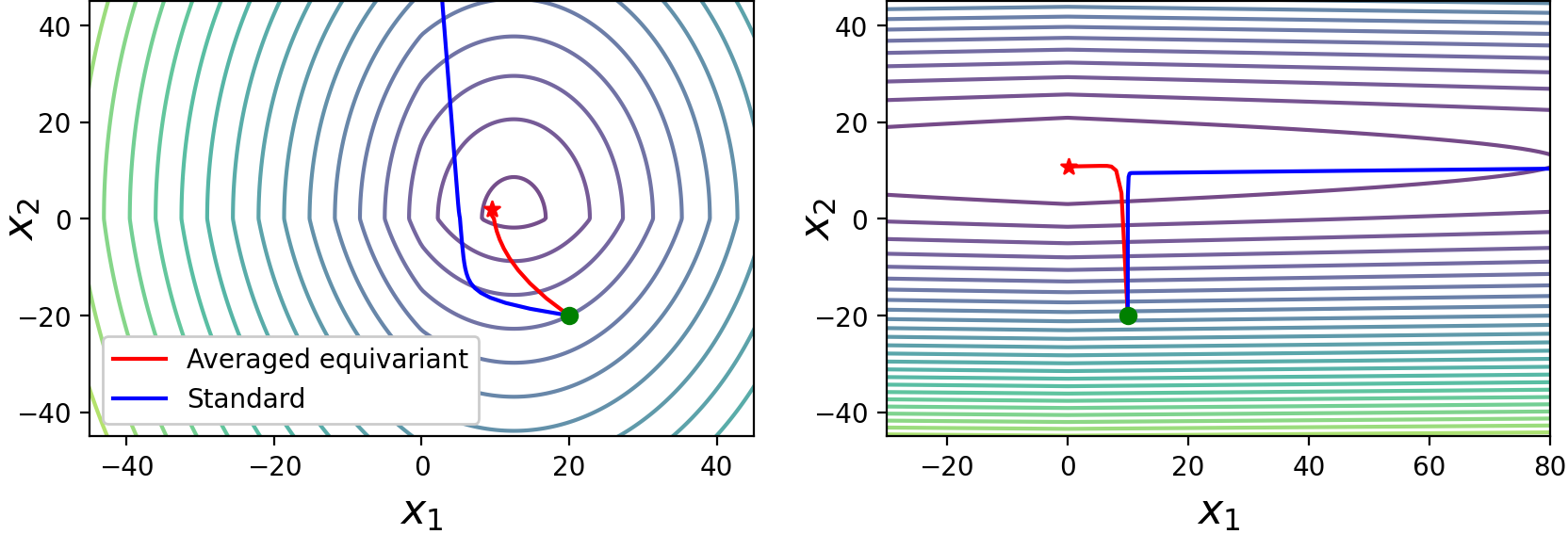}
\end{minipage}
\vspace{-1em}
\caption{Behaviour of the \Cref{eq:pnp_fb} algorithm with an approximated proximity operator (blue curve) and its equivariant counterpart (red curve). Contour lines show the loss in~\cref{eq:min_pb} with $ r(x)=\|B_1x\|_1$. Stars denote the limit point of each sequence (when it exists) and green dots show the initialization points.}
\label{fig:toy_exp}
\vspace{-1.5em}
\end{figure}

\begin{table*}[htbp]
\centering
\footnotesize
\resizebox{0.7\textwidth}{!}{
\begin{tabular}{l@{\hskip 15pt}cccccc}
\toprule
& \multicolumn{2}{c}{\textbf{Deblur (motion)}} & \multicolumn{2}{c}{\textbf{Deblur (Gaussian)}} & \multicolumn{2}{c}{\textbf{MRI}} \\
 & Set3C & BSD10 & Set3C & BSD10 & $\times$4 & $\times$8 \\
 \cmidrule(lr){2-3}
\cmidrule(lr){4-5}
\cmidrule(lr){6-7}
LipDnCNN \cite{pesquet2021learning} & $31.7 \pm 2.0$ & $30.9 \pm 0.7$ & $32.1 \pm 2.5$ & $32.6 \pm 0.9$ & $30.6 \pm 2.6$ & $26.4 \pm 2.3$\\
Eq. LipDnCNN & $31.8 \pm 2.0$ & $31.1 \pm 0.6$ & $32.1 \pm 2.5$ & $32.7 \pm 0.9$ & $30.7 \pm 2.7$ & $26.6 \pm 2.2$\\
 \cmidrule(lr){2-3}
\cmidrule(lr){4-5}
\cmidrule(lr){6-7}
DnCNN & $29.9 \pm 1.9$ & $30.4 \pm 0.3$ & $22.0 \pm 6.0$ & $29.6 \pm 4.6$ & $28.6 \pm 4.3$ & div.\\
Eq. DnCNN & $30.7 \pm 1.6$ & $30.9 \pm 0.1$ & $29.8 \pm 3.1$ & $33.0 \pm 0.6$ & $30.1 \pm 3.9$ & $24.6 \pm 2.9$\\
 \cmidrule(lr){2-3}
\cmidrule(lr){4-5}
\cmidrule(lr){6-7}
DRUNet & $10.1 \pm 1.5$ & $16.5 \pm 9.8$ & $14.6 \pm 8.8$ & $17.5 \pm 10.5$ & $27.7 \pm 4.0$ & $16.9 \pm 8.6$\\ 
Eq. DRUNet & $18.0 \pm 6.0$ & $28.4 \pm 2.3$ & $25.3 \pm 10.1$ & $31.9 \pm 1.1$ & $30.7 \pm 3.0$ & $22.0 \pm 6.0$ \\
\hline
wavelets & $29.6 \pm 2.0$ & $27.2 \pm 1.0$ & $31.2 \pm 2.5$ & $30.9 \pm 0.8$ & $28.6 \pm 2.1$ & $25.4 \pm 1.8$\ \\
TGV & $29.2 \pm 2.1$ & $26.7 \pm 1.3$ & $30.7 \pm 2.5$ & $30.5 \pm 0.9$ & $28.5 \pm 2.3$ & $24.8 \pm 2.1$\\
GSPnP \cite{hurault2021gradient} & $34.6\pm 0.2$ & $33.5 \pm 1.9$ & $35.1\pm 0.7$ & $31.4 \pm 3.3$ & n/a & n/a\\
DPIR \cite{zhang2021plug} & $33.9 \pm 1.9$ & $35.0 \pm 0.5$ & $33.0 \pm 2.5$ & $34.6 \pm 0.7$ & $28.4 \pm 2.2$ & $25.1 \pm 1.9$ \\
\bottomrule
\end{tabular}
}
\vspace{-0.6em}
\caption{Mean reconstruction PSNR for various image restoration problems using the \cref{eq:pnp_fb} algorithm with different backbone denoisers. The first, third, and fifth rows present results for non-equivariant denoisers (LipDnCNN, DnCNN, and DRUNet, respectively), while their equivariant counterparts (Eq. LipDnCNN, Eq. DnCNN, and Eq. DRUNet) are shown in the second, fourth, and sixth rows. The bottom four rows offer benchmarks with standard reconstruction methods. The notation "div." indicates cases where the method diverged.}
\label{table:metrics}
\vspace{-1.6em}
\end{table*}

\section{Experimental results}

In this section, we evaluate the influence of the proposed equivariant approach for different algorithms and linear inverse imaging problems. Our implementation\footnote{available at \href{https://github.com/matthieutrs/EquivariantPnP}{https://github.com/matthieutrs/EquivariantPnP}.} relies on the DeepInverse library \cite{deepinv}.

\subsection{Problems considered}

\paragraph{Image deblurring and image super-resolution} In this setting, we set $y=h*x+\epsilon$ in \eqref{eq:inv_pb} where $h$ is a convolutional kernel and $*$ the circular convolution. We consider either Gaussian deblurring, in which case $h$ is a Gaussian kernel of standard deviation 1, and motion deblurring, in which case $h$ is the first kernel from \cite{levin2009understanding}. Unless specified otherwise, $\epsilon$ is a Gaussian noise with standard deviation $0.01$.
In the image super-resolution (SR) setting, \eqref{eq:inv_pb} writes as $y = \left(h*x \right)_{\Downarrow S}+\epsilon$ where $h$ is a Gaussian kernel of standard deviation 1, and $S$ denotes the undersampling factor. When $S=2$ (resp. $S=4$), $\epsilon$ is a Gaussian noise with standard deviation $0.01$ (resp. $0.05$). We test the proposed method on the Set3C dataset as well as the BSD10 dataset, a subset of 10 images from the BSD68 dataset.

\paragraph{MRI} In this setting, we consider $y=  MFx$ in \eqref{eq:inv_pb} where $M$ is a binary mask and $F$ the 2D Fourier transform. Following \cite{zbontar2018fastmri}, we consider the $\times$4 and $\times$8 acceleration factors. In contrast to the previous problem, no noise is added to the measurements in this case. We test the method on a subset of 10 images from the validation set of the fully acquired k-space data of \cite{zbontar2018fastmri}.

\paragraph{Algorithms and backbone denoisers} We consider several backbone pretrained denoiser, namely DRUNet \cite{zhang2021plug}, SCUNet \cite{zhang2022practical}, SwinIR \cite{liang2021swinir}, DiffUNet \cite{ho2020denoising}, DnCNN \cite{zhang2017beyond} as well as its 1-Lipschitz version (LipDnCNN) from \cite{pesquet2021learning} and the gradient-step denoiser GSNet \cite{hurault2021gradient}. These architectures are representative of state-of-the-art image reconstruction architectures, involving both convolutional and attention layers. 
We stress that none of these networks is equivariant to translations or rotations, to the exception of DnCNN that show approximate translation equivariance.
The choice of the backbone denoiser may influence the chosen algorithm. For instance, the DnCNN and SwinIR denoisers on which this article relies are trained for a fixed level of noise, limiting the ability to finetune the algorithm. 
Moreover, DiffUNet and SCUNet architectures can only be applied to color images, preventing their usage on the MRI problem. 
Unless specified, each algorithm is run for $10^4$ iterations.

\paragraph{Baselines} In this paper, we mainly investigate the influence of equivariant priors on the stability of PnP algorithms. We thus compare the proposed approach with variational approaches which can be seen as a class of convergent PnP algorithms. In particular, we use wavelet denoisers (i.e. $\D(x)=\operatorname{prox}_{\lambda \| \Psi \cdot \|_1}(x)$ for $\Psi$ a redundant wavelet dictionary), as well as total generalized variation (TGV) denoisers \cite{bredies2010total}.
The \Cref{eq:pnp_fb} algorithm with Lipschitz denoiser corresponds to the method from \cite{pesquet2021learning} which ensures convergence of the \Cref{eq:pnp_fb} algorithm.
We also compare our approach with the state-of-the-art DPIR algorithm \cite{zhang2021plug} which runs a small number of steps of a half-quadratic splitting algorithm with fine-tuned decaying stepsizes. As a consequence, DPIR can be seen as representative of non-convergent, fine-tuned PnP algorithm.
We also compare our results with the gradient-step PnP algorithm \cite{hurault2021gradient}, leveraging a nonconvex implicit prior and backtracking and that is representative of state-of-the-art convergent PnP algorithms.

\subsection{Stability of PnP and RED algorithms}
The instability of PnP algorithms often translates into unrealistic artifacts in the reconstructed image. Figure~\ref{fig:gaus_problem} illustrates this phenomenon and shows that the proposed approach allows us to circumvent this drawback.

\begin{figure}[t]
\footnotesize
\centering
\setlength{\tabcolsep}{1pt}%
\begin{tabular}{c c c c}
 Groundtruth & $A^\top y$ & Standard & Equivariant \\
 \includegraphics[width=0.115\textwidth]{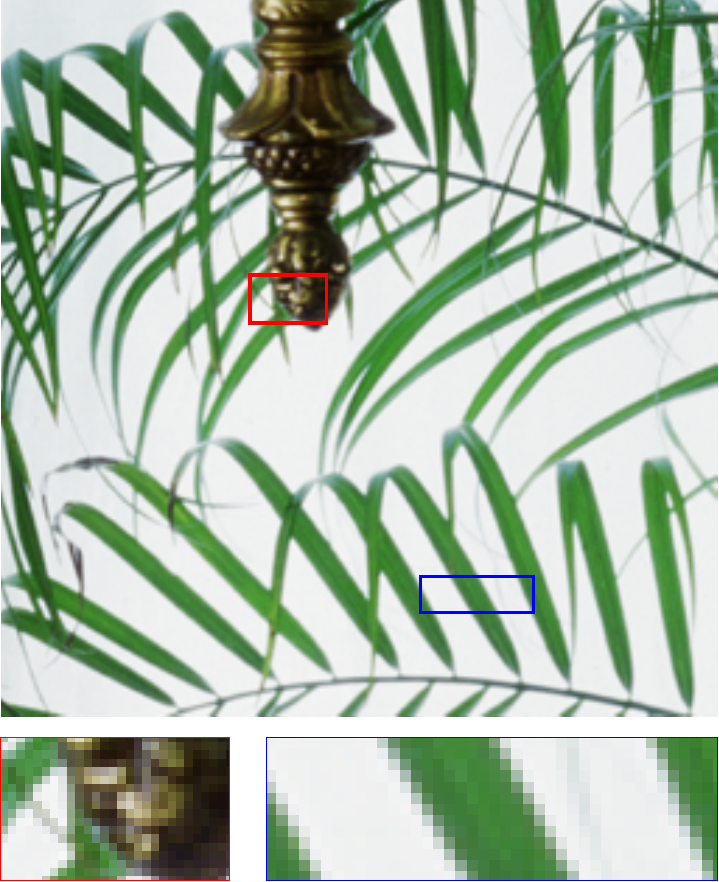} &
 \includegraphics[width=0.115\textwidth]{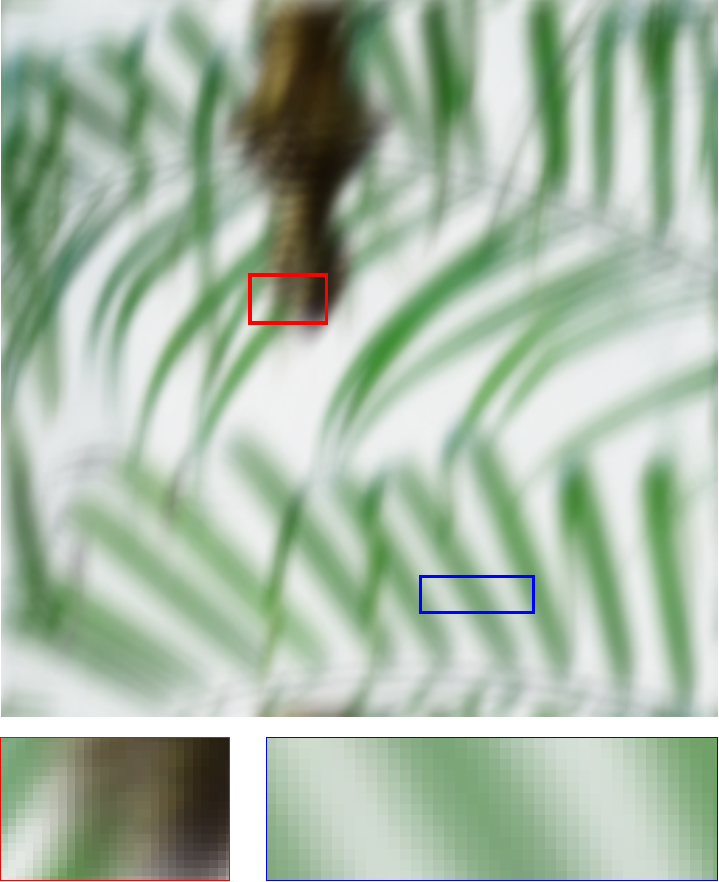} &
 \includegraphics[width=0.115\textwidth]{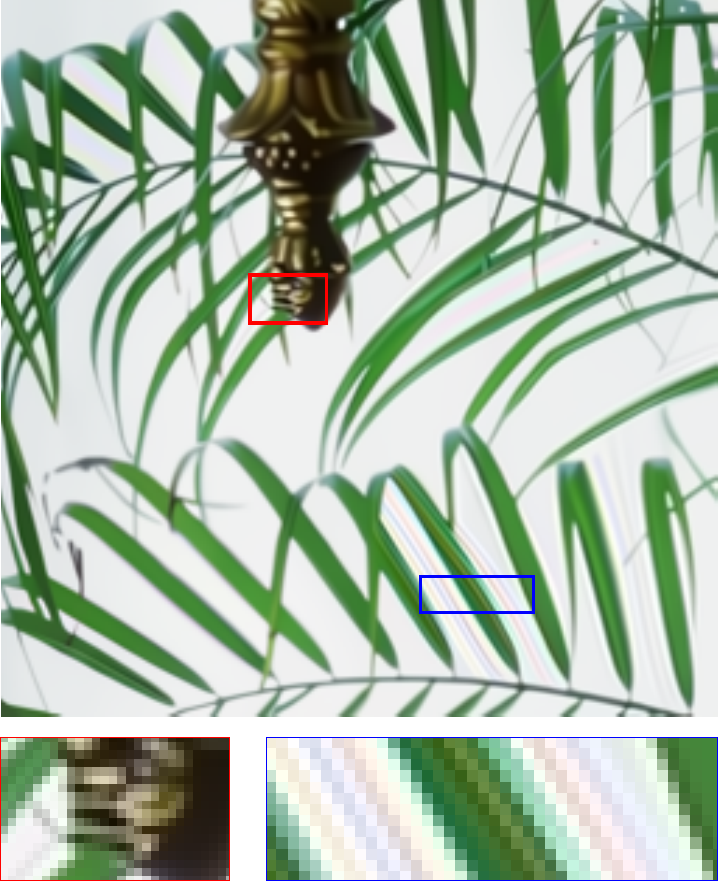} &
 \includegraphics[width=0.115\textwidth]{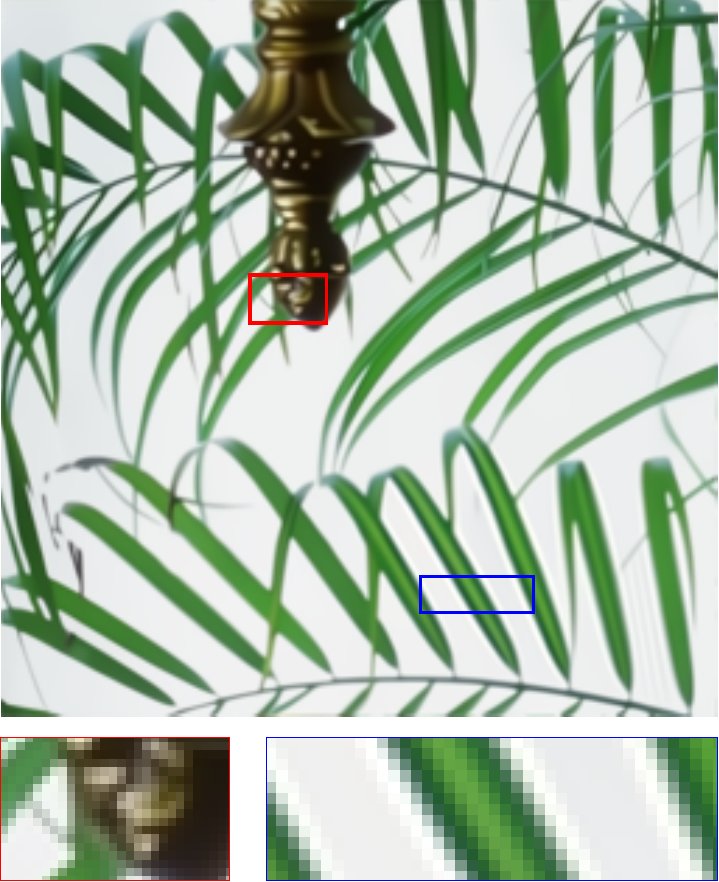}\\
  PSNR & 16.65 & 30.30 & 31.65
\end{tabular}
\vspace{-1.5em}
\caption{Motion deblurring on a Set3C sample with \eqref{eq:pnp_fb} relying on a DRUNet backbone denoiser.}
\vspace{-2em}
\label{fig:gaus_problem}
\end{figure}

\begin{figure*}[t]
\footnotesize
\begin{minipage}{0.49\textwidth}
    \centering
    \setlength{\tabcolsep}{1pt}%
    \begin{tabular}{c c}
    \raisebox{ %
 1.5\height %
 }{\rotatebox[origin=c]{90}{DnCNN}} 
 &
  \includegraphics[width=0.95\textwidth]{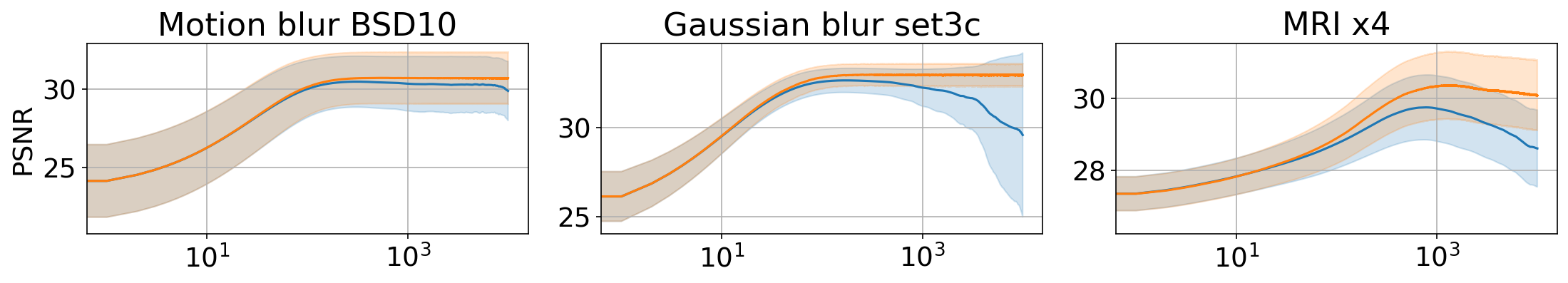}  \\
     \raisebox{ %
 1.5\height %
 }{\rotatebox[origin=c]{90}{DRUNet}} & \includegraphics[width=0.95\textwidth]{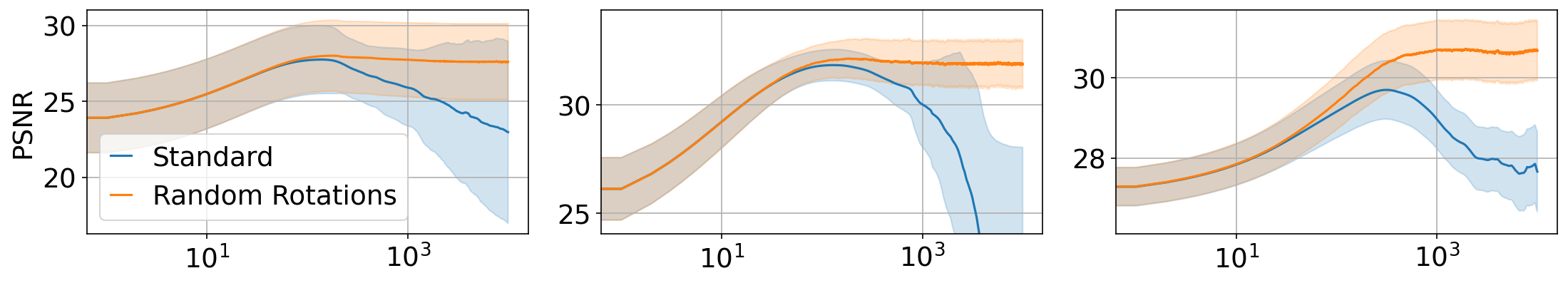}  \\
     \raisebox{ %
 1.\height %
 }{\rotatebox[origin=c]{90}{LipDnCNN}} & \includegraphics[width=0.95\textwidth]{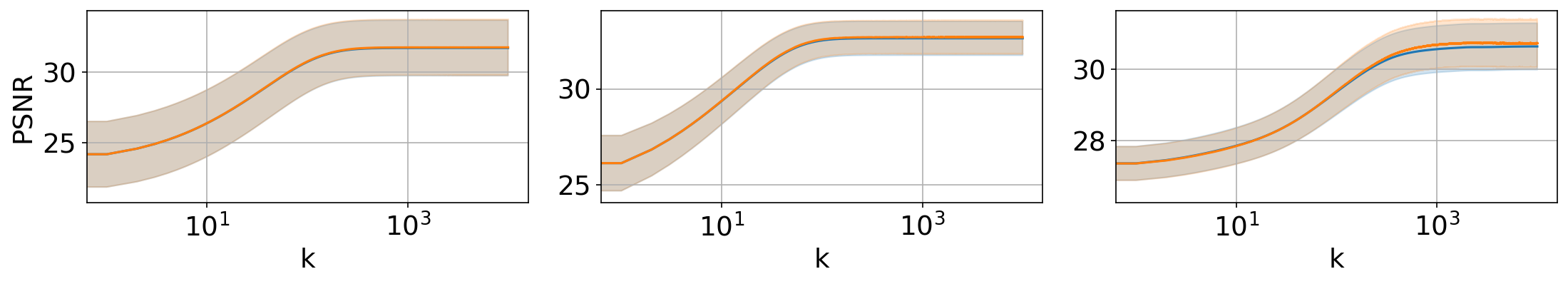}  \\
\end{tabular}
\end{minipage}%
\begin{minipage}{0.49\textwidth}
    \centering
    \setlength{\tabcolsep}{1pt}%
    \begin{tabular}{c c}
 &
  \includegraphics[width=0.95\textwidth]{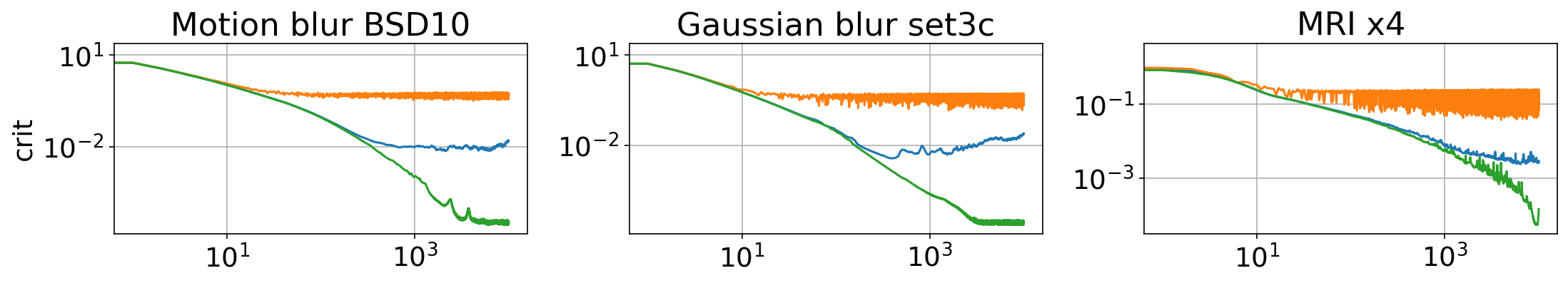}  \\
 & \includegraphics[width=0.95\textwidth]{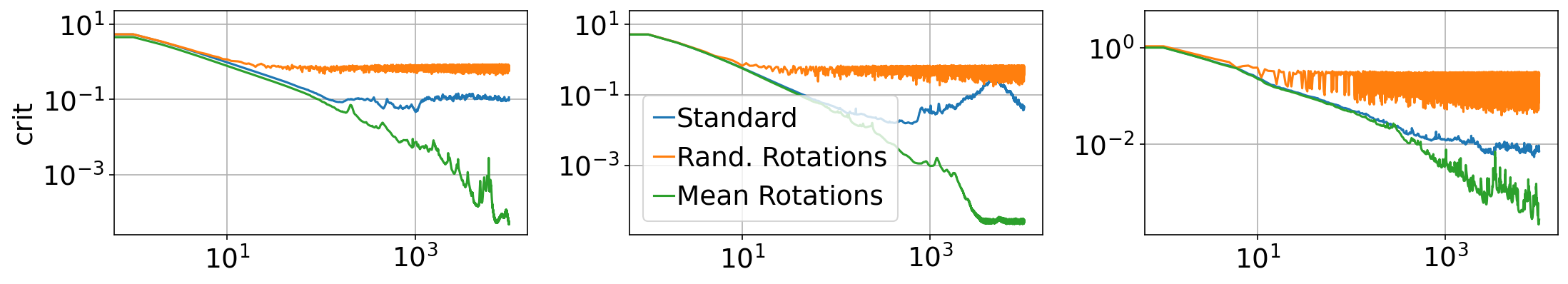}  \\
 & \includegraphics[width=0.95\textwidth]{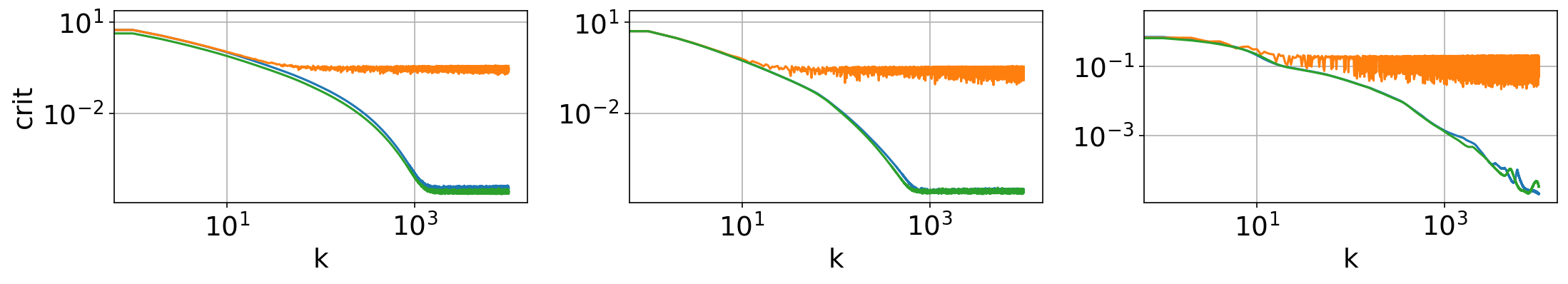}  \\
\end{tabular}
\end{minipage}
\vspace{-1.5em}
\caption{Average PSNR (left) and  convergence criterion $\|x_{k+1}-x_{k}\|/\|x_k\|$ (right) for 3 different imaging problems as a function of the iteration number with different backbone denoisers plugged in the \cref{eq:pnp_fb} algorithm. Top row: DnCNN, middle row: DRUNet, bottom row: 1-Lipschitz DnCNN.}
\vspace{-1.5em}
\label{fig:metrics_crit}
\end{figure*}

Figure~\ref{fig:metrics_crit} shows the PSNR and convergence criterion $\|x_{k+1}-x_{k}\|/\|x_k\|$ on different problems and backbone architecture. For each of these problems, the denoising level $\sigma$ in \eqref{eq:pnp_fb} is set to $\sigma = 0.01$ (resp. $\sigma = 0.015$) for DnCNN (resp. DRUNet) backbone architectures.
We notice that the \eqref{eq:pnp_fb} algorithm with equivariant denoiser shows a more stable PSNR along iterations than its non-equivariant counterpart. In particular, we observe in the right panel of Figure~\ref{fig:metrics_crit} that the convergence rate of the \eqref{eq:pnp_fb} algorithm with $\G$-equivariant DnCNN matches the behaviour of the algorithm with Lipschitz denoisers. Lastly, we notice that the proposed equivariant approach also benefits convergent \eqref{eq:pnp_fb} algorithms relying on Lipschitz backbone denoisers.
We further stress that reconstructions obtained with the equivariant Monte-Carlo estimates are consistent with those obtained with deterministic Reynolds averaging (see \cref{tab:MC_compare_rebuttal}).

In the case of the \eqref{eq:red_gd} algorithm, we observe similar behaviours. Reconstructions on a $\times$2 SR problem are shown in Figure~\ref{fig:red_sr_results} and associated convergence plots can be found in Figure~\ref{fig:red_convergence}. 
In the case of a non-equivariant DRUNet backbone, the reconstructed image shows important geometric artifacts that disappear in the equivariant case.

We however stress that the proposed approach may fail to stabilize the algorithm in certain settings. 
For example, in the case of a SCUNet backbone denoiser, strong artifacts are visible from 50 iterations only in the reconstruction in both the classical and equivariant version of the \eqref{eq:pnp_fb} algorithm as seen in Figure~\ref{fig:gaus_problem_bsd68_rebuttal}. Similarly, we did not observe convergence of the \cref{eq:red_gd} algorithm with the DiffUNet backbone, thus requiring early stopping of the algorithm in order to reach good reconstruction results.

Lastly, several studies have shown that the noise level $\sigma$ plays an important role in the stability of the PnP algorithm. Increasing $\sigma$ may help to solve the instability issue at the cost of over-smoothed reconstructions.

\begin{figure}[t]
\vspace{-0.5em}
\footnotesize
    \centering
    \setlength{\tabcolsep}{1pt}%
    \begin{tabular}{@{\hskip 0pt}c c c c c c c@{\hskip 0pt}}
 & \scalebox{0.8}{Groundtruth} & \scalebox{0.8}{$y$} & \scalebox{0.8}{TGV (28.40)} & \scalebox{0.8}{Wavelets (28.64)} &  \\
 & \includegraphics[width=0.09\textwidth]{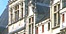} &
 \includegraphics[width=0.09\textwidth]{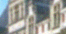} &
 \includegraphics[width=0.09\textwidth]{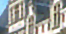} &
 \includegraphics[width=0.09\textwidth]{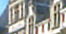} \\
& \scalebox{0.8}{SCUNet (23.54)} & \scalebox{0.8}{SwinIR (27.85)}  & \scalebox{0.8}{DRUNet (9.58)} & \scalebox{0.8}{GSNet (29.22)} & \scalebox{0.8}{DnCNN (29.66)}\\
 \raisebox{ %
 .6\height %
 }{\rotatebox[origin=c]{90}{\scalebox{0.85}{Standard}}}  &
 \includegraphics[width=0.09\textwidth]{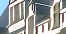} 
 &
 \includegraphics[width=0.09\textwidth]{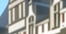} &
 \includegraphics[width=0.09\textwidth]{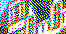} 
 & 
 \includegraphics[width=0.09\textwidth]{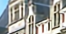} &
 \includegraphics[width=0.09\textwidth]{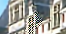}
 \\
 & \scalebox{0.8}{SCUNet (25.14)} & \scalebox{0.8}{SwinIR (27.83)} & \scalebox{0.8}{DRUNet (29.22)} & \scalebox{0.8}{GSNet (29.25)} & \scalebox{0.8}{DnCNN (30.36)}\\

 \raisebox{ %
 .8\height %
 }{\rotatebox[origin=c]{90}{\scalebox{0.85}{Equiv.}}}  &
 \includegraphics[width=0.09\textwidth]{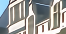} &
 \includegraphics[width=0.09\textwidth]{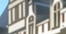}  
 &
 \includegraphics[width=0.09\textwidth]{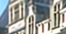} 
 & 
 \includegraphics[width=0.09\textwidth]{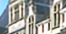} &
 \includegraphics[width=0.09\textwidth]{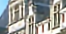} 
 
 \\
\end{tabular}
\vspace{-1em}
\caption{Gaussian deblurring with standard deviation $\sigma = 0.02$ on a BSD10 sample (detail) for different denoising backbone plugged in the (PnP) algorithm.}
\label{fig:gaus_problem_bsd68_rebuttal}
\vspace{-0.5em}
\end{figure}

 \begin{figure}[t]
 \vspace{-0.5em}
\footnotesize
    \centering
    \setlength{\tabcolsep}{1pt}%
    \begin{tabular}{c c c} 
    & DiffUNet & DRUNet \\
 \includegraphics[width=0.32\linewidth]{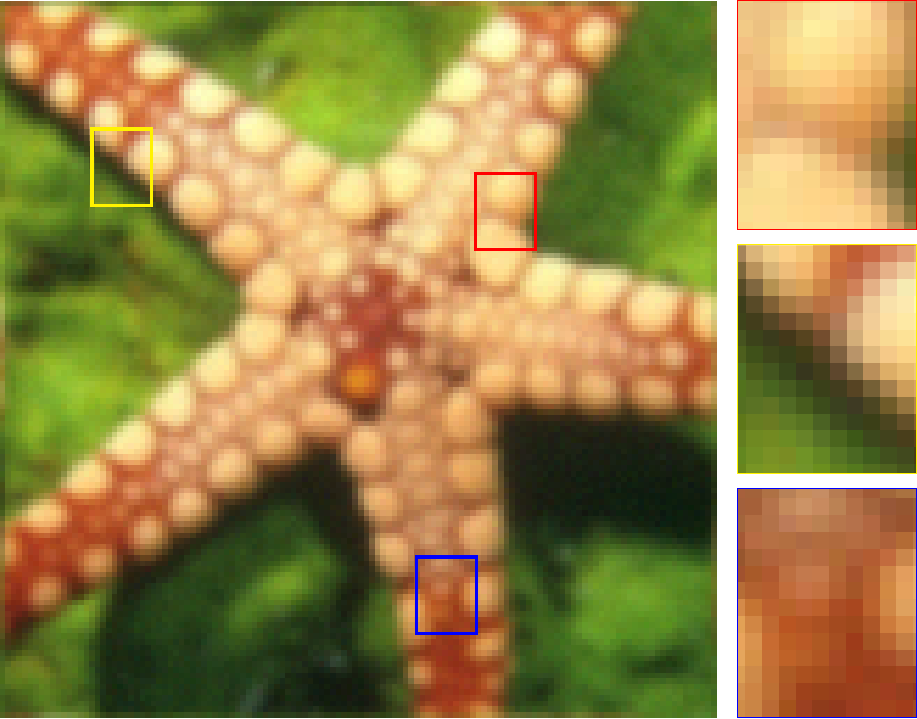} &
 \includegraphics[width=0.32\linewidth]{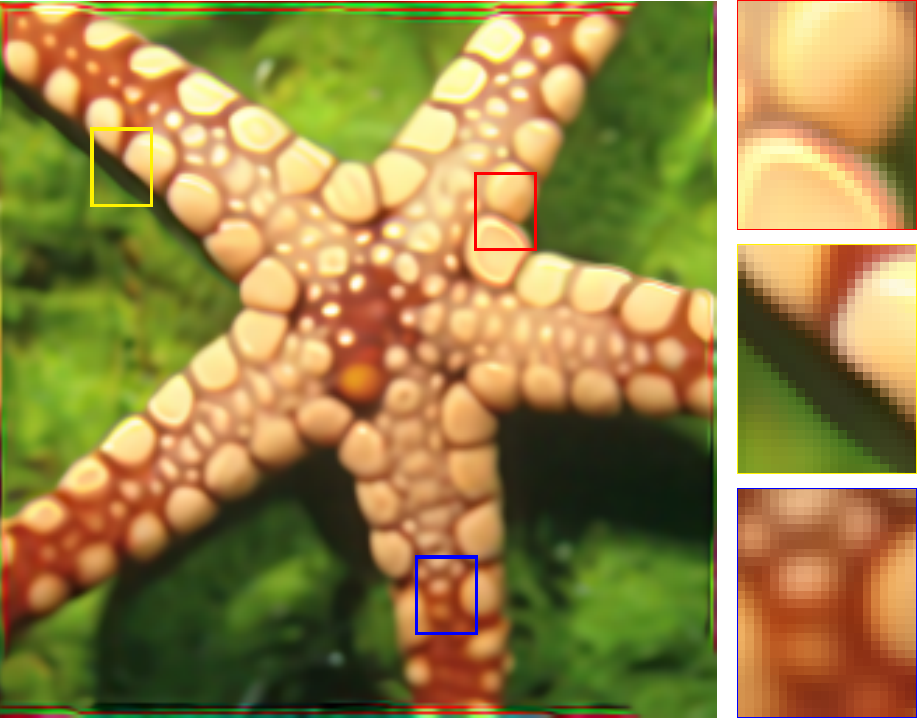} &
 \includegraphics[width=0.32\linewidth]{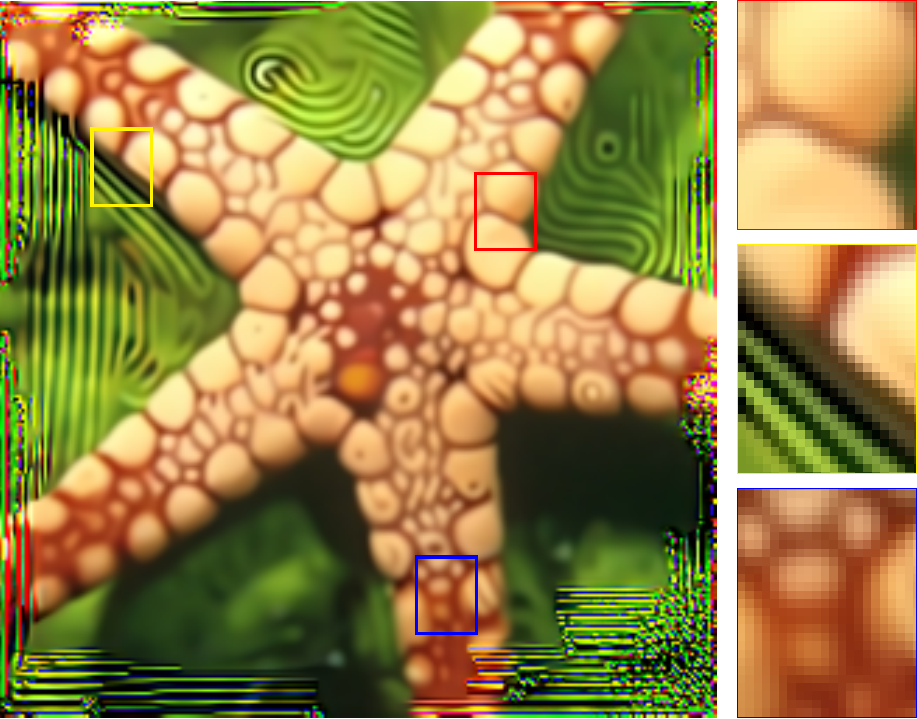} \\
 Observed & Standard (21.55) & Standard (20.75)  \\
 \includegraphics[width=0.32\linewidth]{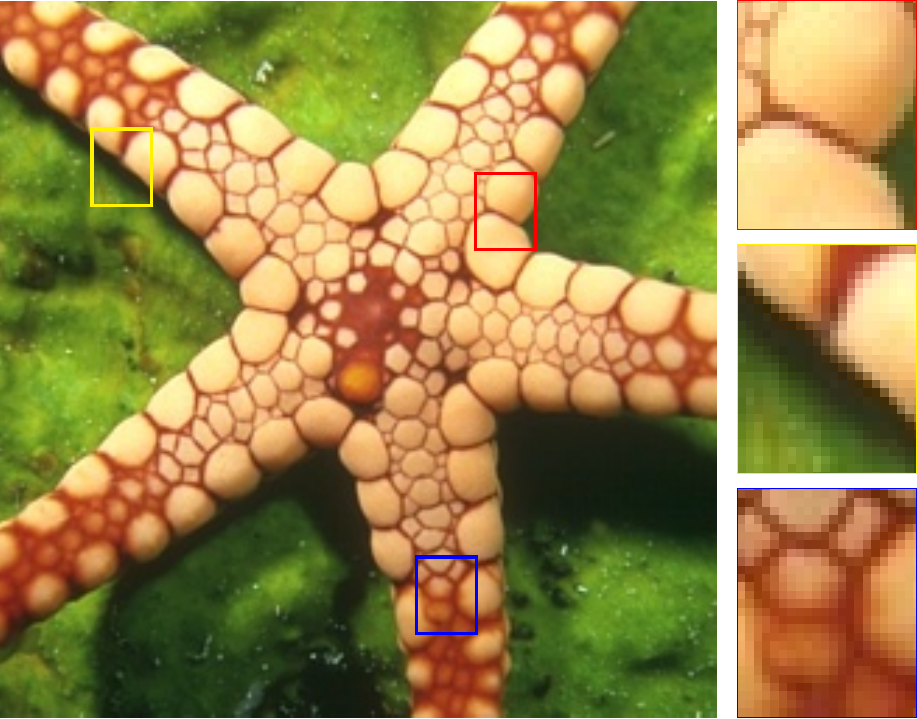} &
 \includegraphics[width=0.32\linewidth]{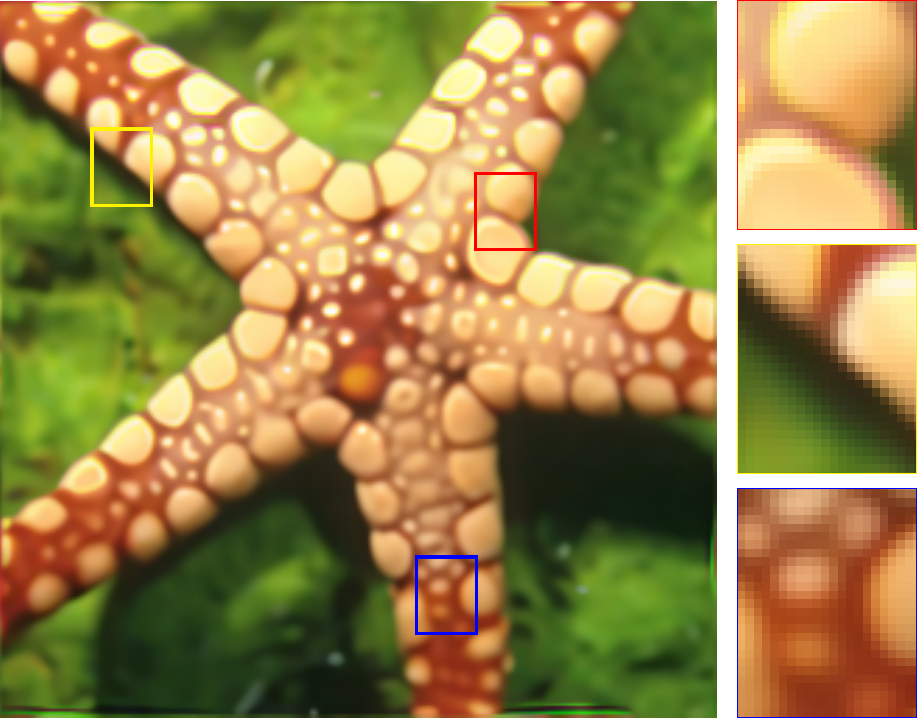} &
 \includegraphics[width=0.32\linewidth]{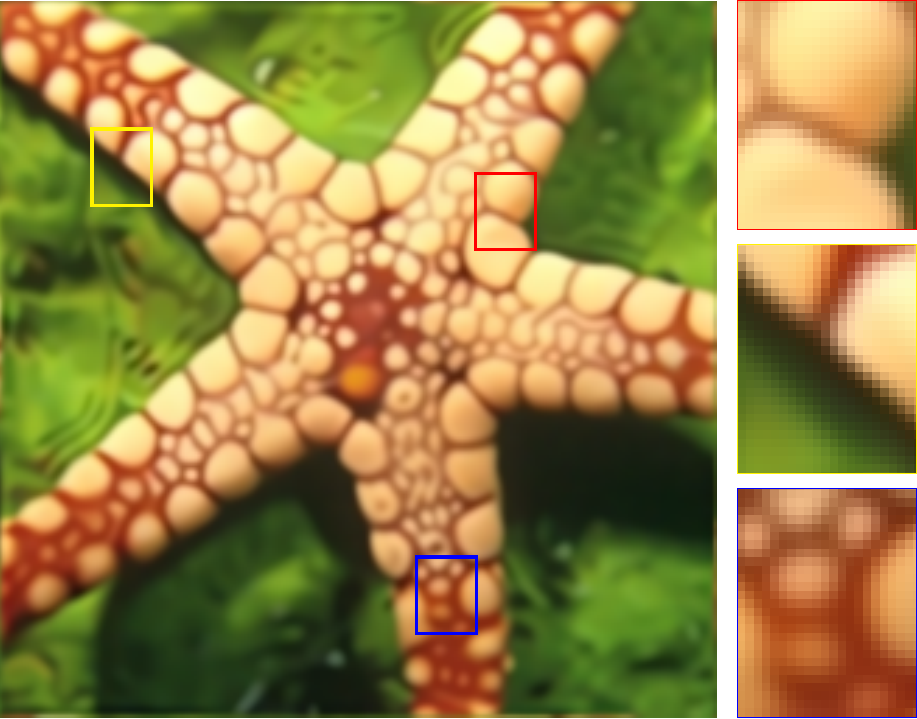} \\
 Groundtruth (PSNR) & Equivariant (22.26) & Equivariant (24.94)  \\
\end{tabular}
\vspace{-1em}
\caption{Results of the \cref{eq:red_gd} algorithm on a $\times 2$ SR problem for different backbone denoisers. Middle column: DiffUNet; right column: DRUNet. Top row: standard algorithm; Bottom row: equivariant algorithm.}
\label{fig:red_sr_results}
\vspace{-2em}
\end{figure}

 \begin{figure}[t]
 \vspace{-1em}
\footnotesize
    \centering
    \setlength{\tabcolsep}{1pt}%
    \begin{tabular}{c c} 
     DiffUNet & DRUNet \\
    \includegraphics[width=0.5\linewidth]{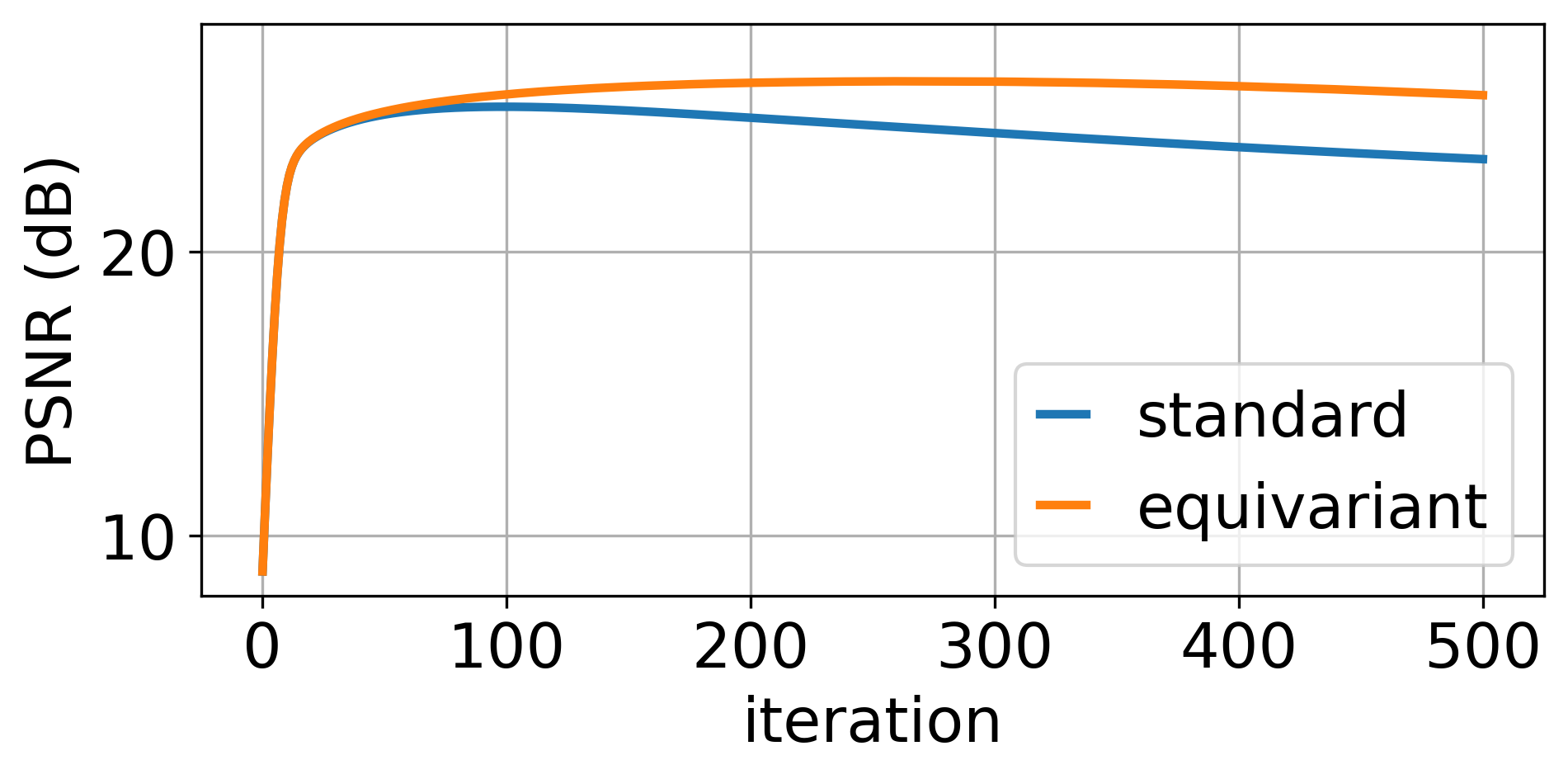} &
 \includegraphics[width=0.5\linewidth]{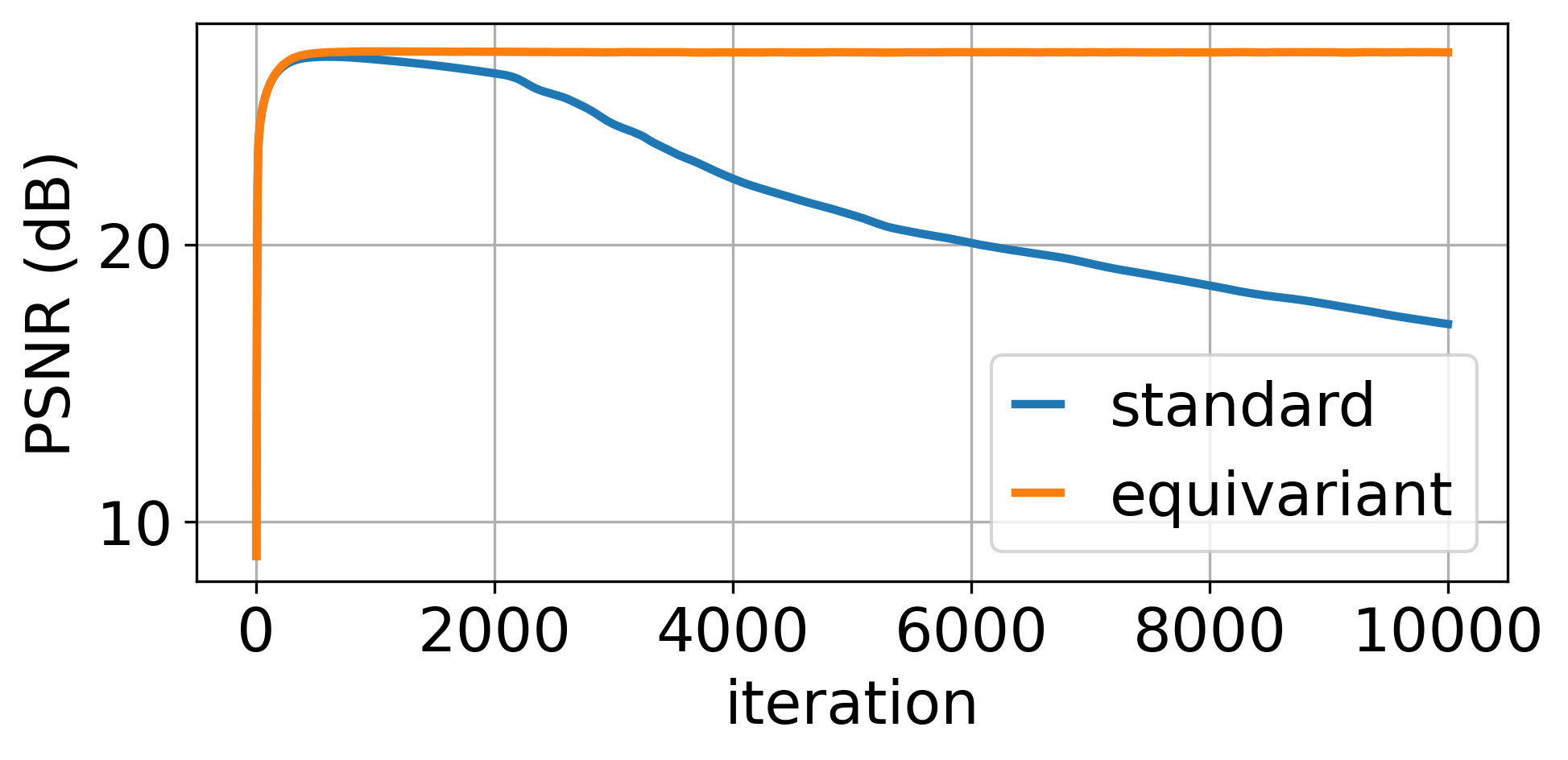}\\
\end{tabular}
\vspace{-1em}
\caption{Evolution of the PSNR along the \cref{eq:red_gd} algorithm for different backbone denoisers associated to the reconstructions shown in Figure~\ref{fig:red_sr_results}.}
\label{fig:red_convergence}
\vspace{-1em}
\end{figure}

\subsection{Interplay with the kernel of $A$}

In the previous section, we have seen that equivariant denoisers can prevent the emergence of artifacts along theiterations of \eqref{eq:pnp_fb}. Despite their unnatural aspect, these artifacts are not incompatible with good data-fidelity measures. In fact, PnP algorithms offer no control over $\ker(A)$, which is nontrivial by nature of the ill-posed inverse problem \eqref{eq:inv_pb}; 
artifacts appearing during the reconstruction are therefore likely to belong to $\ker(A)$.

This phenomenon can be illustrated in the case of MRI, where $\ker(A)$ corresponds to the non-sampled subspace of the Fourier domain.
In this case, $A$ is not rotation equivariant and Proposition \ref{prop:composition} suggests that enforcing equivariance of the prior improves the stability.
This is illustrated in Figure~\ref{fig:mri_128_fact4} showing the reconstruction with and without the proposed equivariant algorithmic update.
In the standard (non-equivariant) setting, mild artifacts appear between iteration $i = 10^3$ and $j = 10^4$. These may appear unnoticed in $x_j$; however, they clearly appear when plotting the difference $x_j-x_i$. Interestingly, the Fourier spectrum of these artifacts shows significantly more energy at frequencies that were not sampled.
The same experiment with the equivariant algorithm shows no such artifacts and a more uniform Fourier spectrum. 

\begin{figure}
[t]
\footnotesize
    \centering
    \setlength{\tabcolsep}{1pt}%
    \begin{tabular}{c c c c }
 & \includegraphics[width=0.14\textwidth]{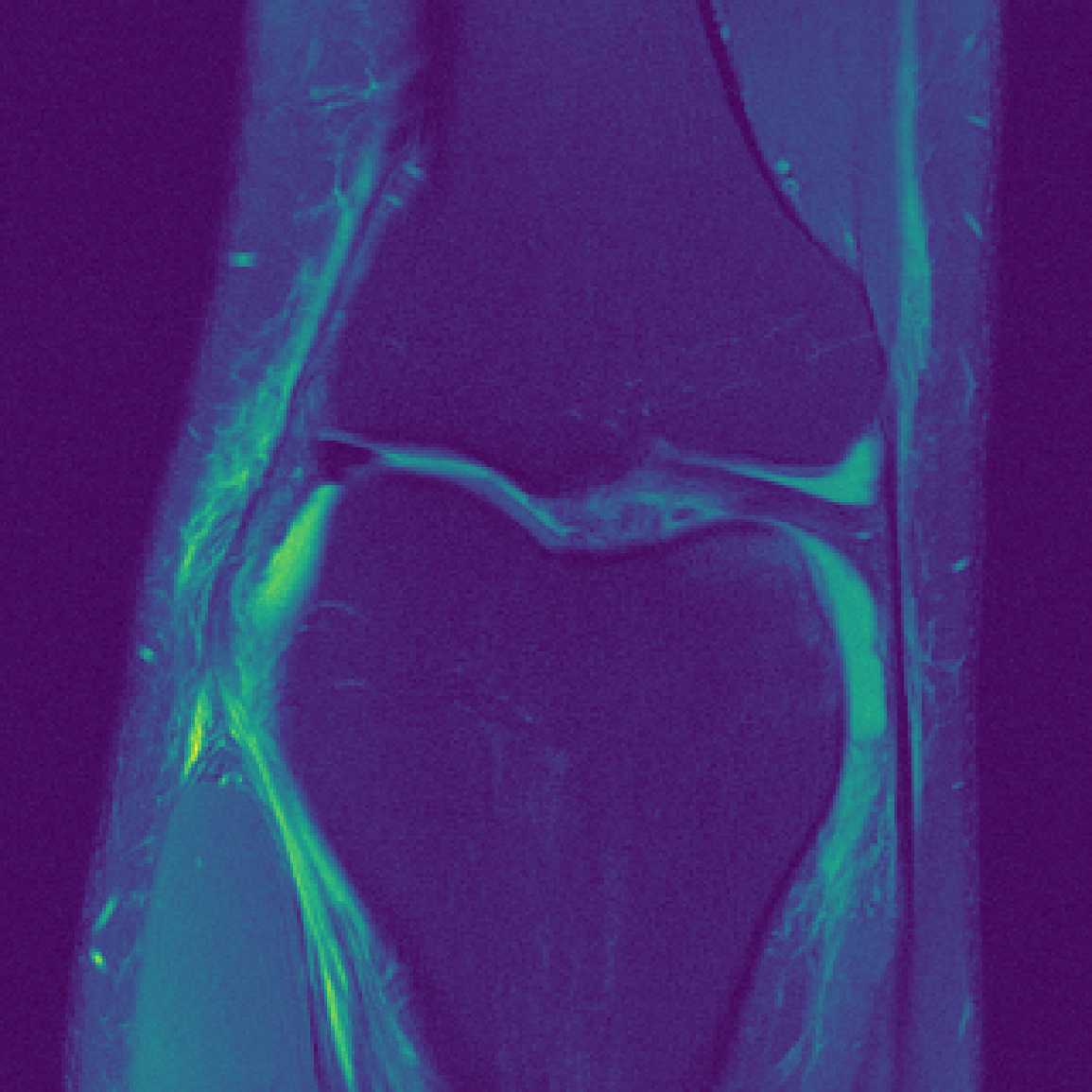} &
 \includegraphics[width=0.14\textwidth]{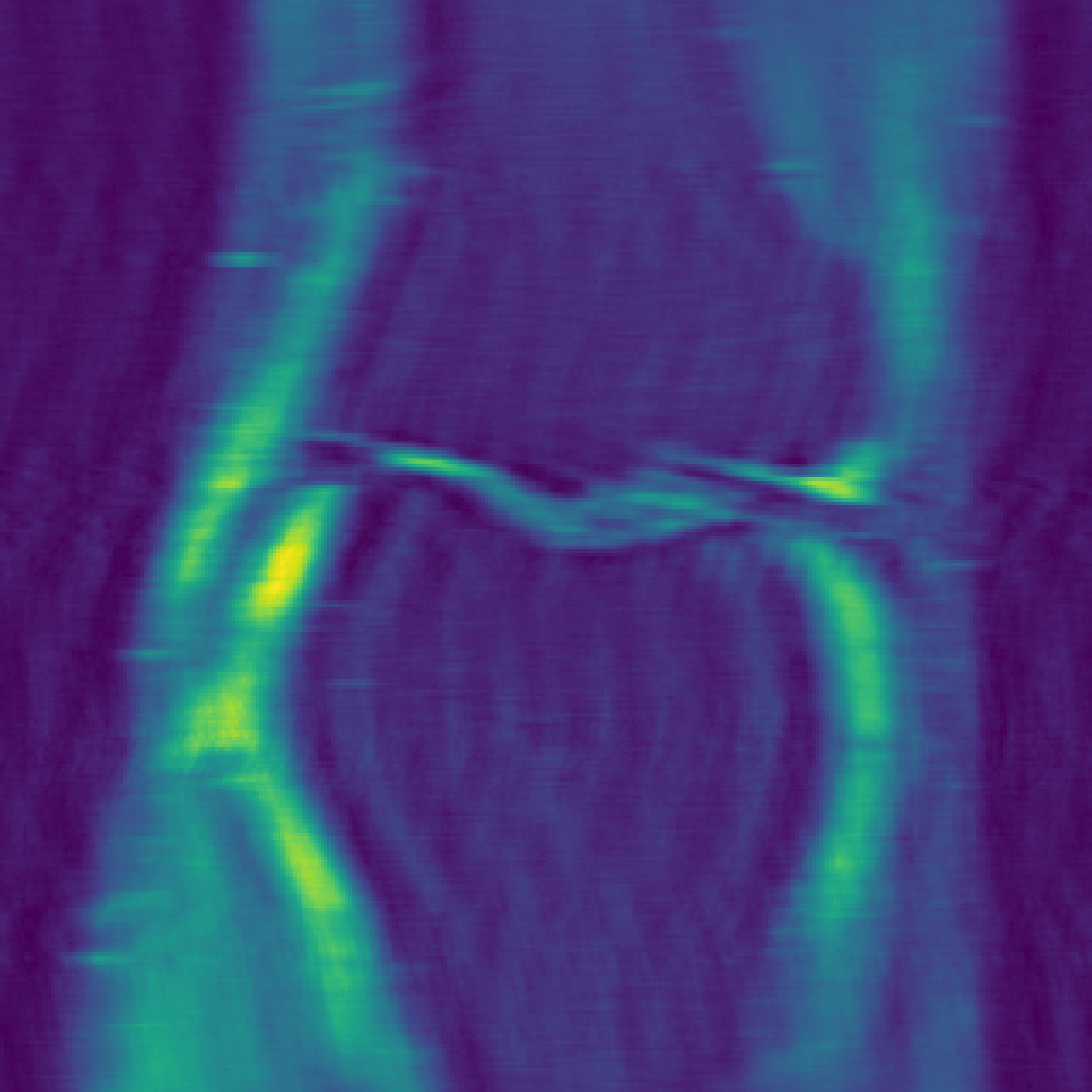} &
 \includegraphics[width=0.14\textwidth]{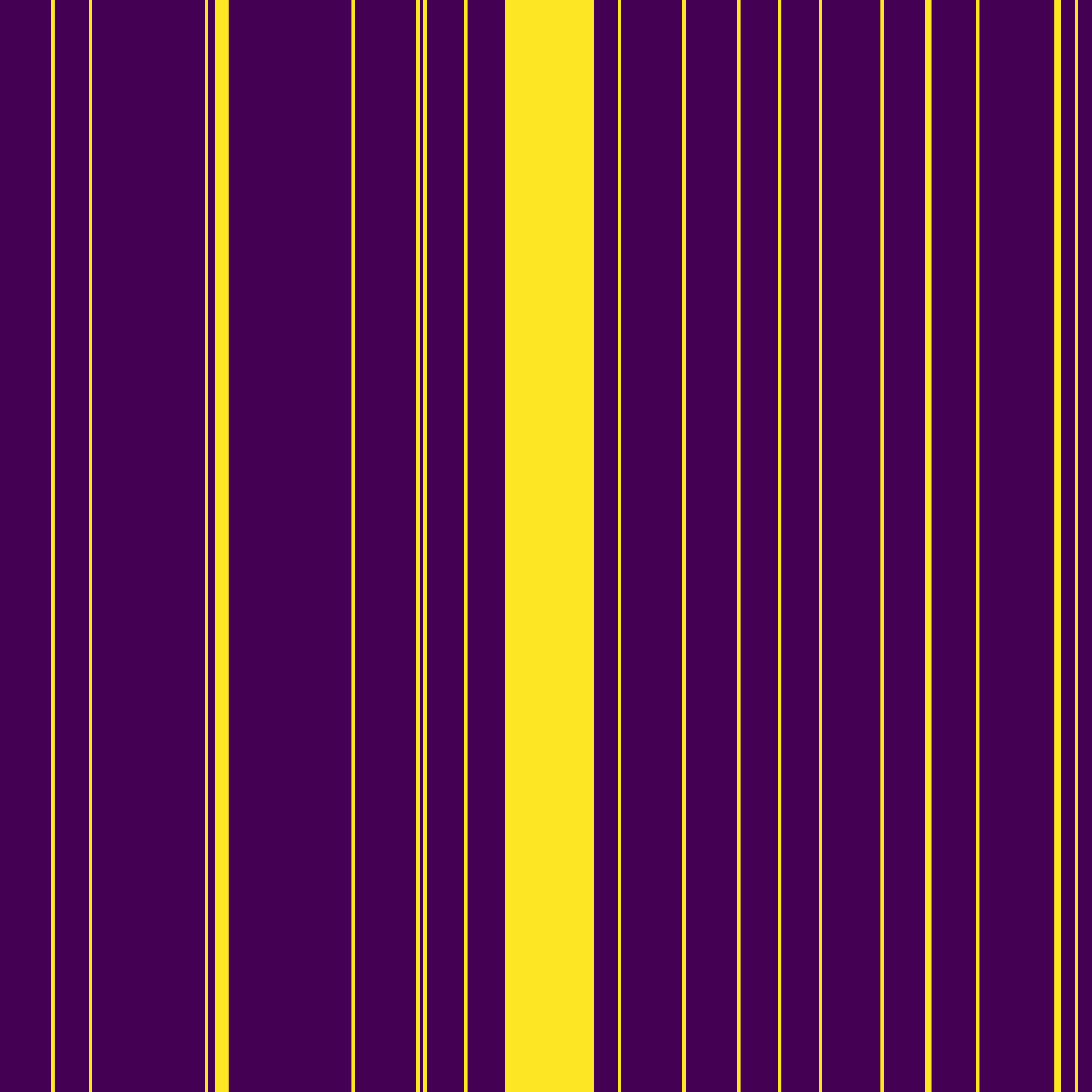} \\
 & Groundtruth & $A^\top y$ & Mask \\
\raisebox{ %
 1.4\height %
 }{\rotatebox[origin=c]{90}{Standard}} & \includegraphics[width=0.14\textwidth]{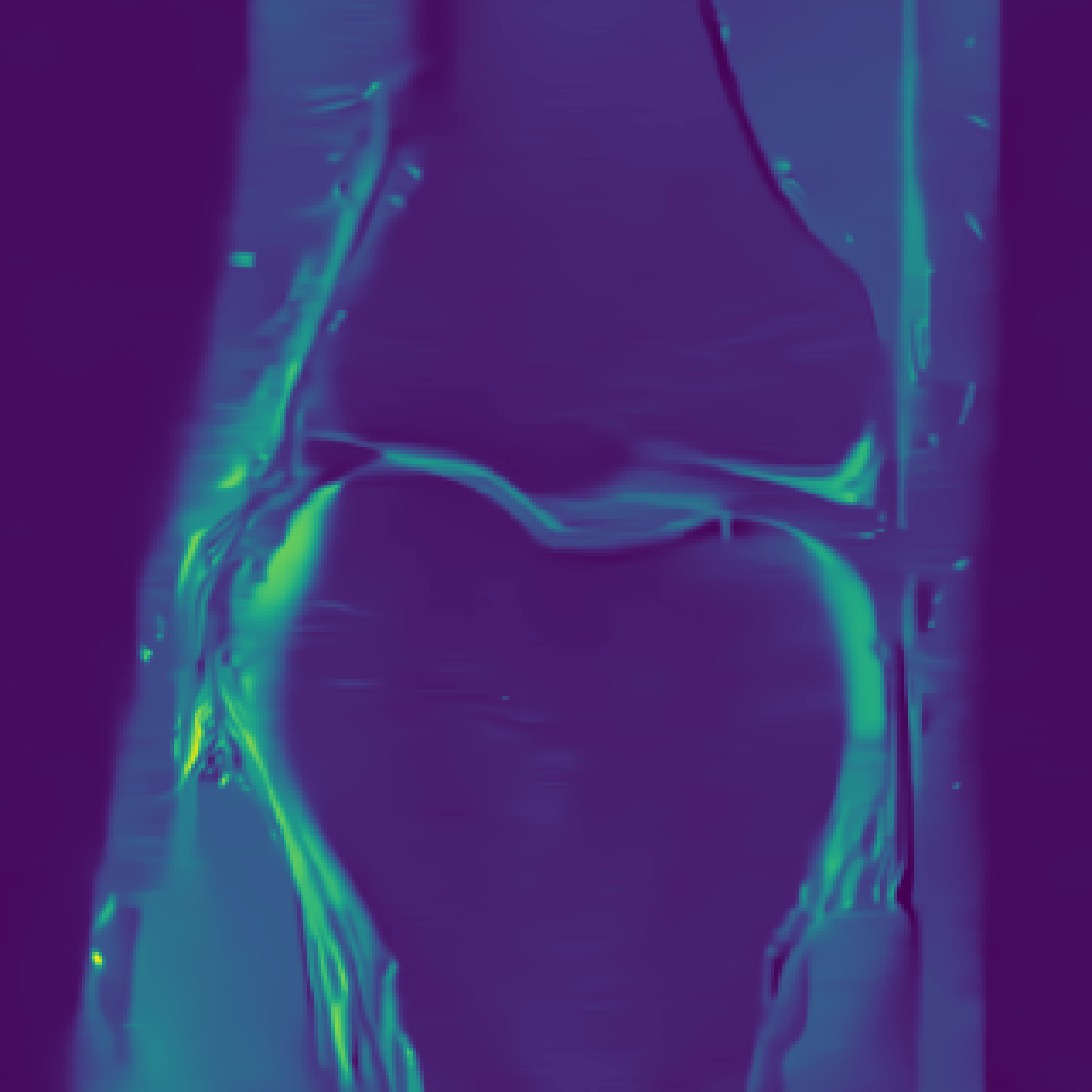} &
 \includegraphics[width=0.14\textwidth]{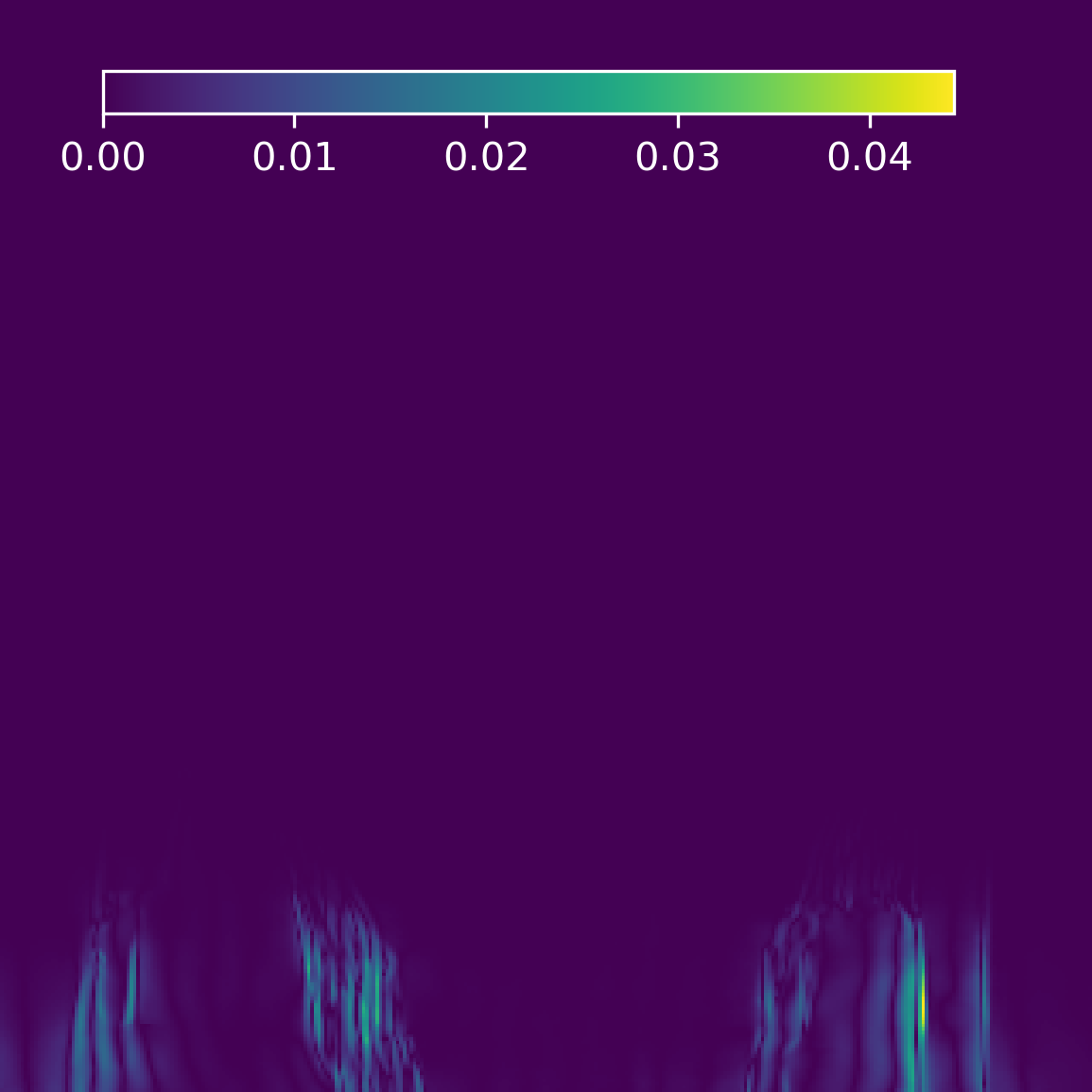} &
 \includegraphics[width=0.14\textwidth]{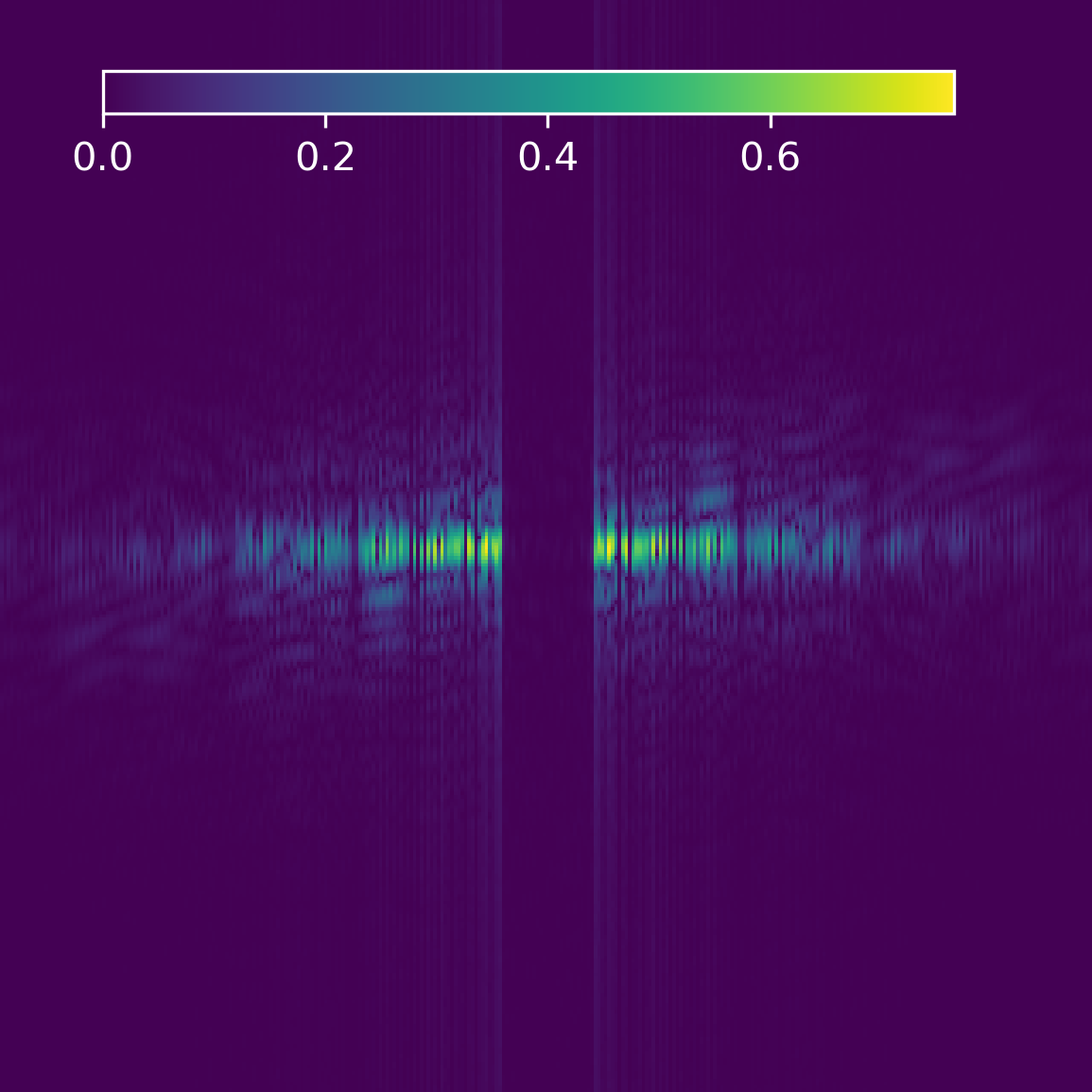} \\
 & $x_{j}$, PSNR=34.10 & $|x_{j}-x_{i}|$ & $\operatorname{FFT}(x_{j}-x_{i})$ \\
\raisebox{ %
 1.4\height %
 }{\rotatebox[origin=c]{90}{Equivariant}} & \includegraphics[width=0.14\textwidth]{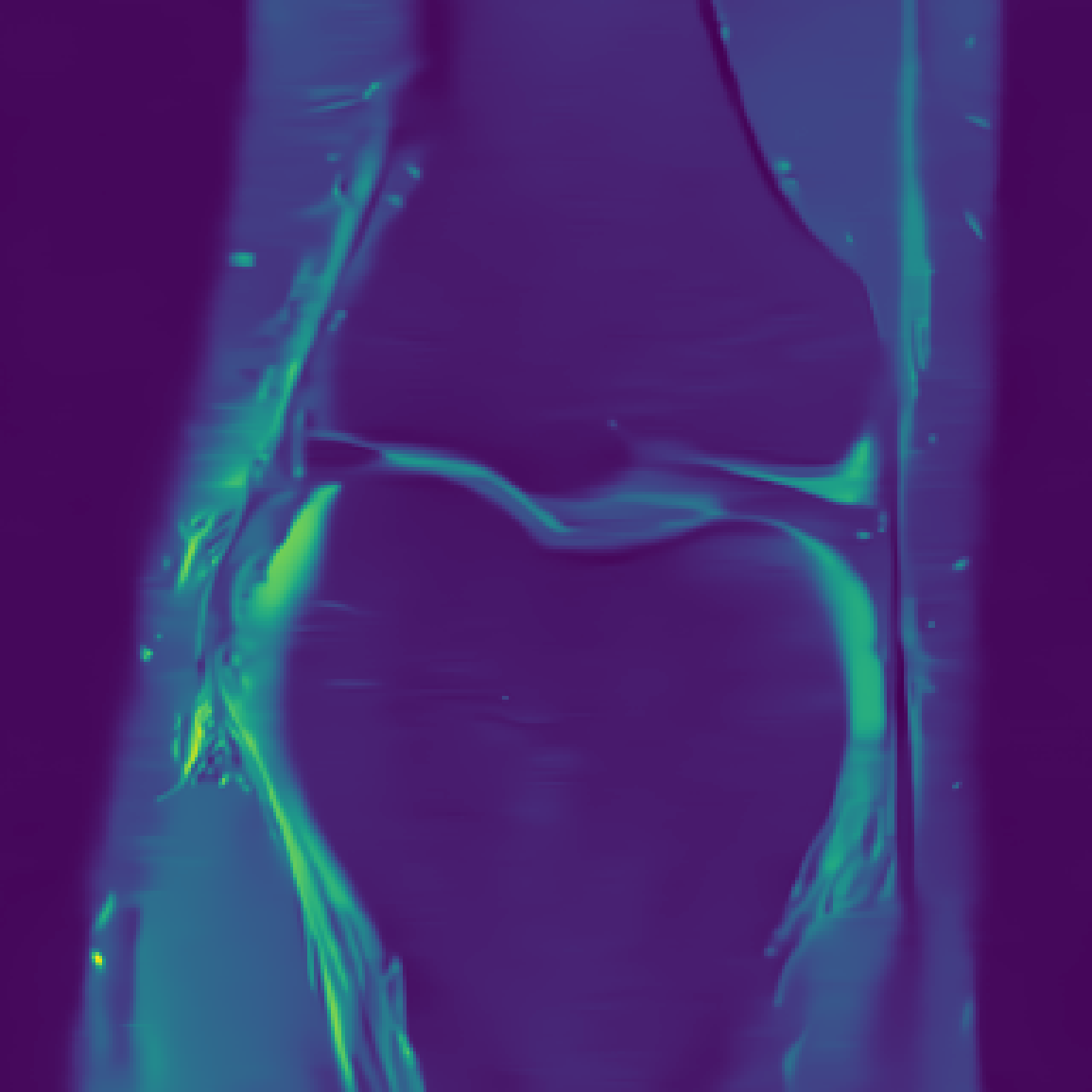} &
 \includegraphics[width=0.14\textwidth]{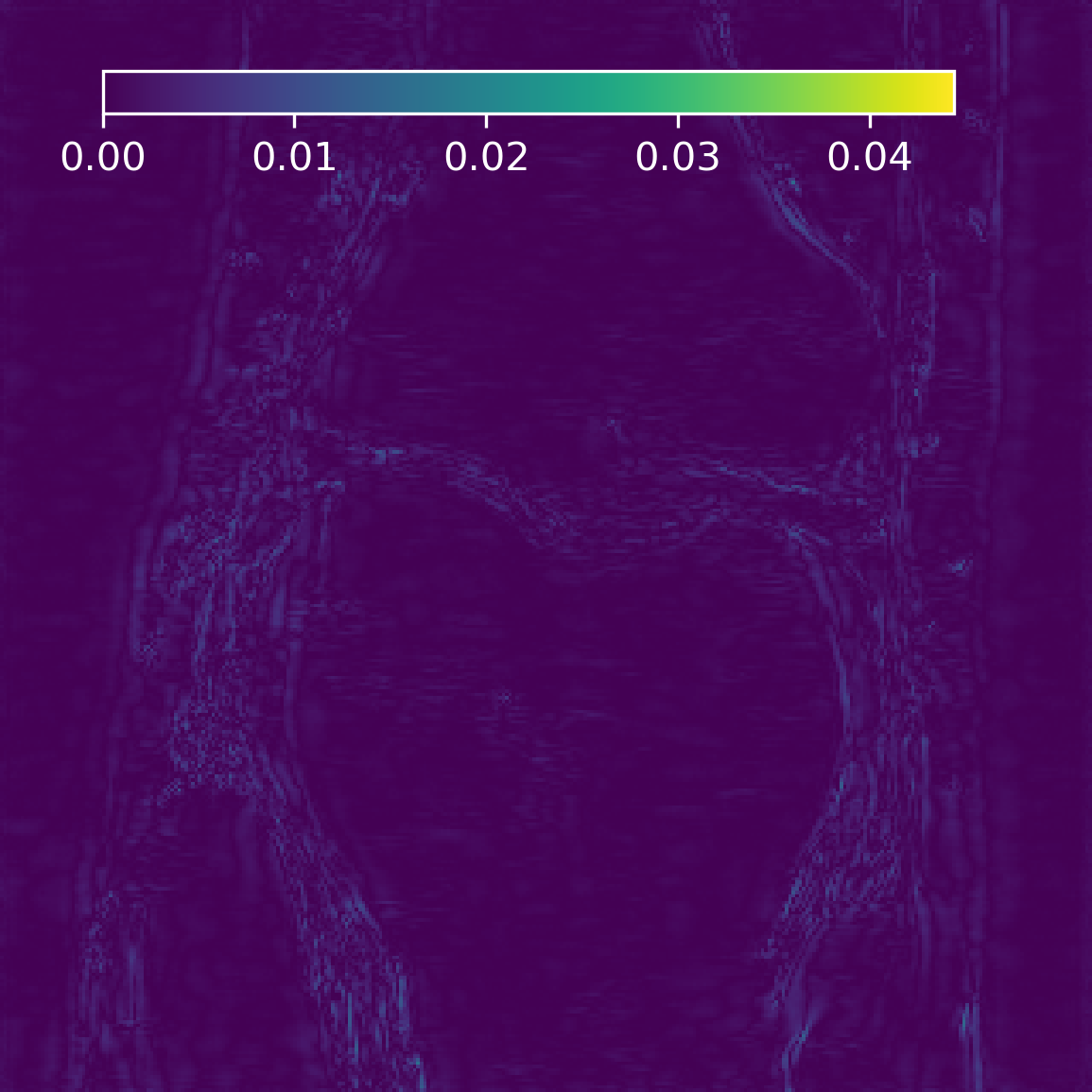} &
 \includegraphics[width=0.14\textwidth]{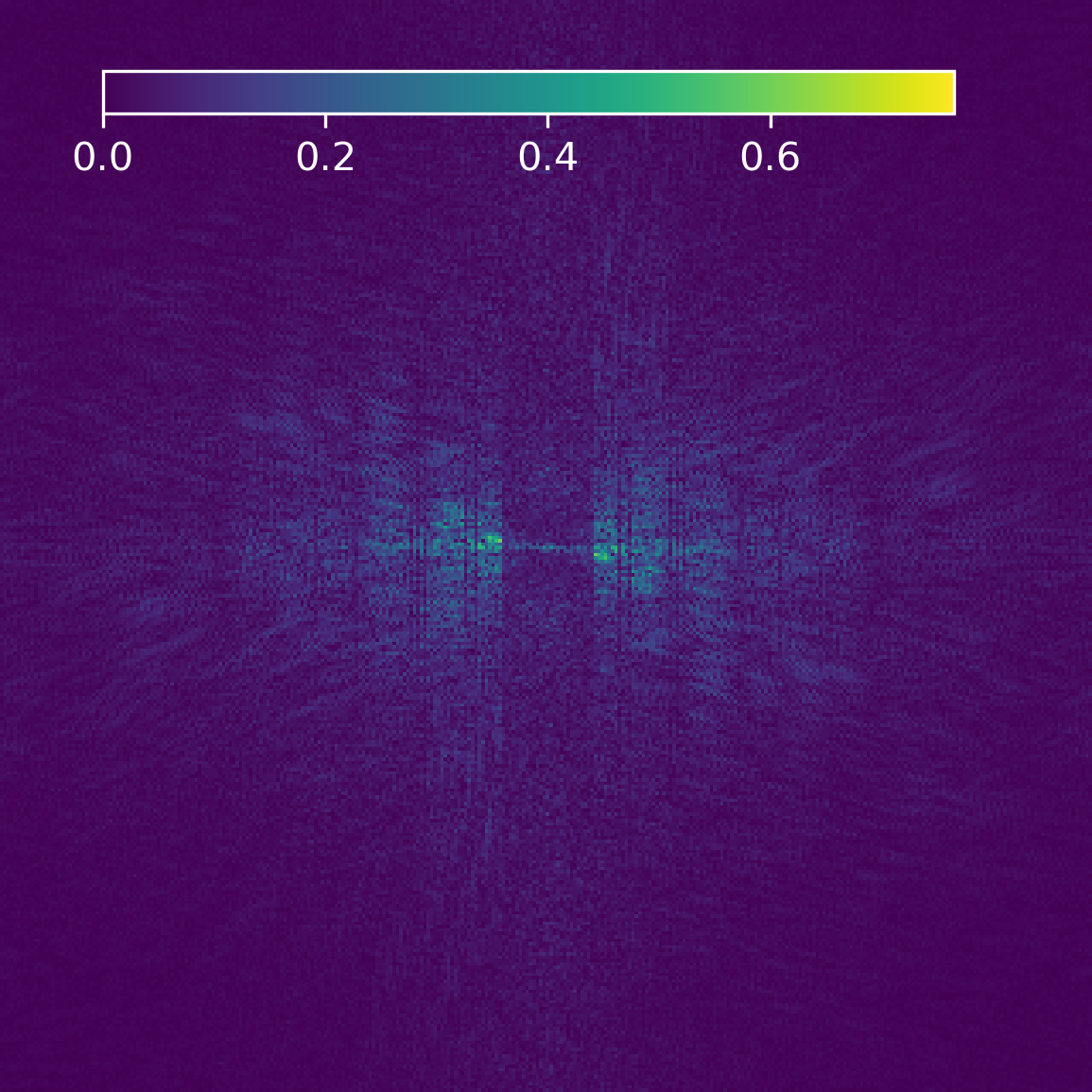} \\
 & $x_{j}$, PSNR=35.26 & $|x_{j}-x_{i}|$ & $\operatorname{FFT}(x_{j}-x_{i})$ \\
\end{tabular}
\vspace{-1em}
\caption{Evolution of the reconstruction along the iterates of \cref{eq:pnp_fb} with a (non 1-Lipschitz) DnCNN backbone. The top row shows the ground truth, back-projected data, and Fourier mask. The second (resp. third) row shows, from left to right: reconstruction at iteration $j=10^4$; difference $|x_j-x_i|$ between reconstructions at iteration $j=10^4$ and at iteration $i=10^3$; the difference in the Fourier domain (displayed in logarithmic scale).}
\label{fig:mri_128_fact4}
\vspace{-1em}
\end{figure}

\subsection{Sampling from RED prior}

The gain in stability provided by the proposed approach opens the door to more robust sampling algorithms relying on implicit denoising priors, such as \cref{eq:ula}, where a sufficiently large number of iterations is required in order to obtain good estimators. We show in Figure~\ref{fig:sampling_motion_ula} the estimated mean and variance obtained with the \Cref{eq:ula} algorithm, with both equivariant and non-equivariant DRUNet backbone denoisers.
In the non-equivariant setting, we observe similar artifacts to those obtained in the deterministic case on both the estimated mean and variance; these artifacts vanish in the equivariant case.

 \begin{figure}
[t]
\footnotesize
    \centering
    \setlength{\tabcolsep}{1pt}%
    \begin{tabular}{c c c} 
    & Standard & Equivariant \\
 \includegraphics[width=0.32\linewidth]{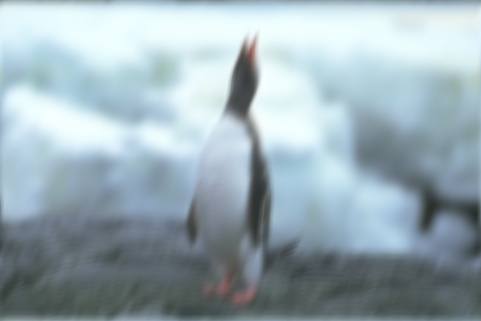} &
 \includegraphics[width=0.32\linewidth]{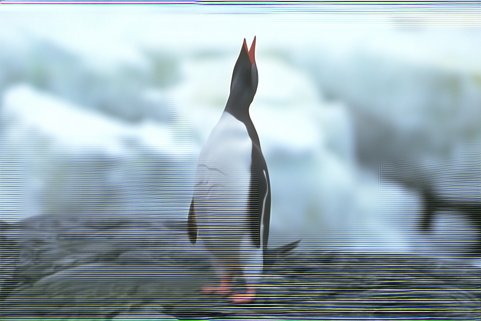} &
 \includegraphics[width=0.32\linewidth]{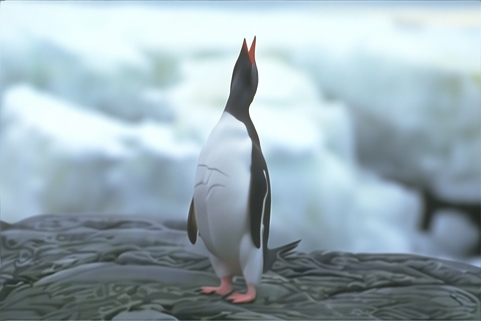} \\
 Observed & Estimated mean & Estimated mean  \\
 \includegraphics[width=0.32\linewidth]{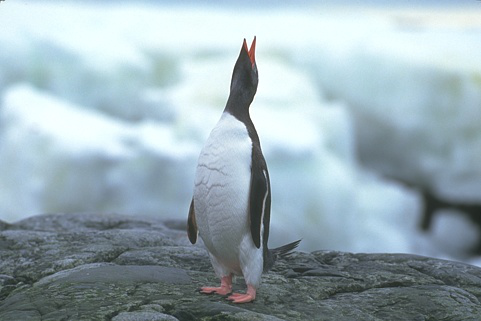} &
 \includegraphics[width=0.32\linewidth]{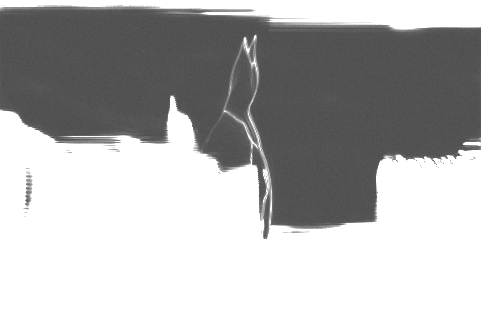} &
 \includegraphics[width=0.32\linewidth]{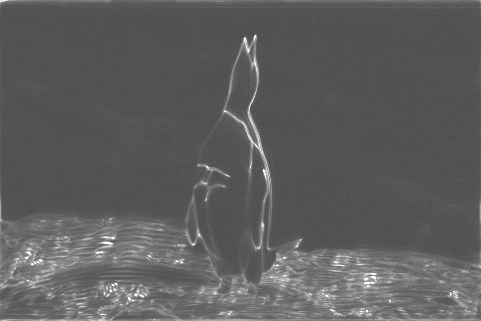} 
 \\
Ground truth & Estimated variance & Estimated variance  \\
\end{tabular}
\vspace{-1em}
\caption{Sampling posterior mean and variance of a motion deblurring problem on a BSD10 sample for a DRUNet backbone plugged in the \cref{eq:ula} algorithm.}
\label{fig:sampling_motion_ula}
\vspace{-2em}
\end{figure}

The leftmost plot of Figure~\ref{fig:ula_sampling_metrics} further illustrates this phenomenon. We notice that after a few hundred iterations, the reconstruction quality with the non-equivariant \Cref{eq:ula} algorithm collapses, yielding an irrelevant MCMC chain. 
While enforcing equivariance improves the situation, increasing the noise level in the denoiser in \Cref{eq:pnp_fb} can significantly improve the situation, see the rightmost plot of Figure~\ref{fig:ula_sampling_metrics}.

\begin{figure}
\footnotesize
    \centering
    \includegraphics[width=0.95\linewidth]{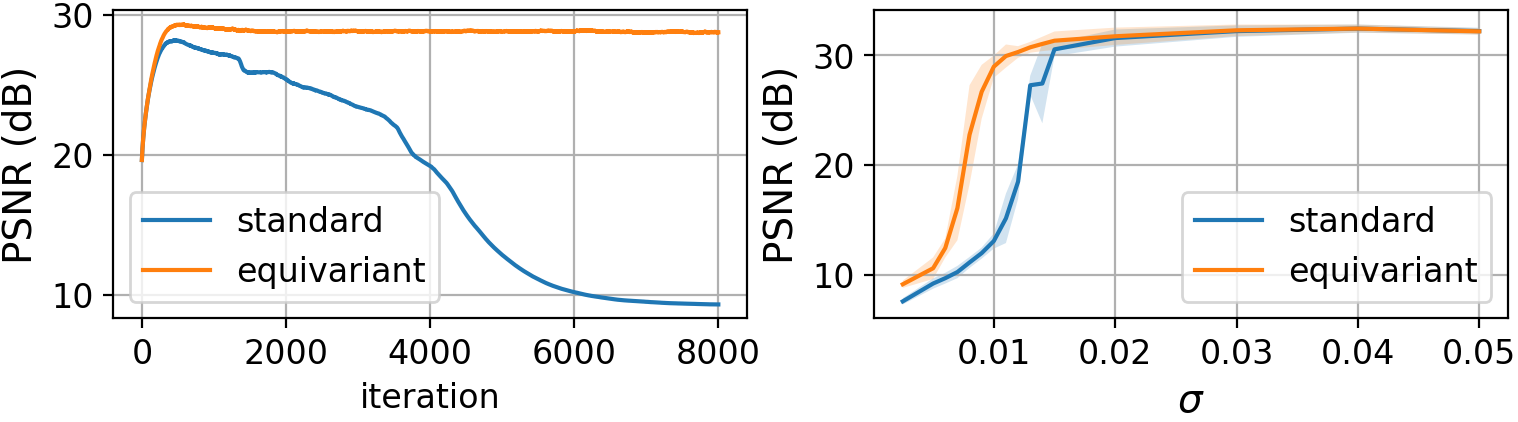}
    \vspace{-1em}
    \caption{Sampling with the \Cref{eq:ula} algorithm on a motion deblurring and the DRUNet prior, in the standard and equivariant cases. Left: PSNR of a sample from the MCMC chain for a sample of the Set3C dataset. Right: average PSNR over the Set3C dataset as a function of $\sigma$ in the denoising prior.}
    \label{fig:ula_sampling_metrics}
\vspace{-2em}
\end{figure}

\section{Limitations}

While our experiments show that the proposed method improves both the stability of the algorithmic scheme and the reconstructed image quality over its non-equivariant counterpart, we stress that it remains prone to divergence and/or hallucinating artifacts. 
This can for example be seen in Figure~\ref{fig:metrics_crit} in the case of the DRUNet backbone, which shows decreasing mean PSNR after a certain number of iterations, or in Figure~\ref{fig:gaus_problem_bsd68_rebuttal} for the SCUNet backbone, which shows important reconstruction artifacts.
In particular, the proposed equivariant approach does not clearly improve over its non-equivariant counterpart when using a SCUNet backbone.

\section{Conclusion}

In this work, we proposed a simple yet efficient method for ensuring approximate equivariance of implicit denoising priors. 
In essence, the method amounts to sampling and applying a group action at random at each step of the algorithm.
In spite of its simplicity, this method shows interesting theoretical properties. In the case of linear denoisers, for example, it allows us to enforce the symmetry of the Jacobian, which is a cornerstone property when aiming to derive explicit priors from implicit denoising priors. 
Furthermore, it can be shown that the Lipschitz constant of the equivariant linear denoiser can only be lower than that of its non-equivariant counterpart, thus improving the stability of the resulting PnP algorithm. 
We showcase the symmetrization effect of equivariance on the Jacobian of several architectures, and illustrate its stabilization effect for three families of imaging algorithms, namely PnP, RED, and ULA. Importantly, this stabilization procedure is not detrimental to the reconstruction quality, as often observed in the literature. 
However, despite the significant improvements brought by equivariance in terms of stability and image reconstruction quality, the proposed method remains prone to divergence and artifact contamination in reconstructions. 

\subsection*{Acknowledgements}
This work was supported by the BrAIN grant (ANR-20-CHIA-0016) and was granted access to the HPC resources of IDRIS under the allocation 2023-AD011014344 made by GENCI. J.~T.~is supported by the ANR's JCJC grant UNLIP.

{
\small
\bibliographystyle{ieeenat_fullname}
\bibliography{main}
}

\appendix

\section{Details on the non-linear example}

\paragraph{Derivation of $D_{\G}$}

We have
\begin{equation}
\begin{aligned} 
D_\G(x) &= \frac{1}{|\G|} \sum T_g^{-1} B_2 \operatorname{prox}_{\gamma \lambda \|\cdot\|_1}(B_1T_g x)\\
        &= \frac{1}{|\G|} \sum T_g^{-1} (B_1 + P) T_g  \operatorname{prox}_{\gamma \lambda \|\cdot\|_1}( B_1x)\\
        &= \left(B_1 +  \frac{1}{|\G|}\sum T_g^{-1} PT_g\right)  \operatorname{prox}_{\gamma \lambda \|\cdot\|_1}( B_1x)
\end{aligned}
\end{equation}
yielding the desired result.
The second step uses the fact that $B_2 = B_1 + P$ and that $B_1$ and $\operatorname{prox}$ are $\G$-equivariant functions. The third step just uses that $B_1$ is a $\G$-equivariant function.

\paragraph{Numerical details for Figure~\ref{fig:toy_exp}} For both the leftmost and rightmost examples, we consider the group $\mathcal{G}$ consisting of permuations of the coordinates of the vectors. This is a group with a single element $g$, the matrix representation of its linear application being the unitary matrix
\begin{equation}
    T_g = \begin{pmatrix}
        0 & 1 \\
        1 & 0
    \end{pmatrix}.
\end{equation}

In the leftmost example, we use $A = \operatorname{diag}(2, 1)$, $B_1=I$ ($B_1$ is thus $\mathcal{G}$-equivariant) and $\lambda=10$. The perturbation and its $\mathcal{G}$-average are 
\begin{equation*}
    P = \begin{pmatrix}
    -0.228 & -0.023 \\
    0.066 &  0.1
    \end{pmatrix}
\quad
    P_{\mathcal{G}} = \begin{pmatrix}
    -0.064 &   0.022 \\
    0.022 & -0.064
    \end{pmatrix},
\end{equation*}
with associated norms $\|P\|_F = 0.26$, $\|P_{\mathcal{G}}\|_F = 0.10$. The \cref{eq:pnp_fb} algorithm is ran with $\gamma = 5e-2$.

In the rightmost example, we use $A = \operatorname{diag}(2, 5e-4)$, $B_1=I$ and $\lambda=2$. The perturbation and its $\mathcal{G}$-average are 
\begin{equation*}
    P = \begin{pmatrix}
    0.0275 & 0.0244\\
    0.0112 & -0.1842
    \end{pmatrix}
,
    P_{\mathcal{G}} = \begin{pmatrix}
    -0.0783 &  0.0178 \\
    0.0178 & -0.0783
    \end{pmatrix},
\end{equation*}
with associated norms $\|P\|_F = 0.0469$, $\|P_{\mathcal{G}}\|_F = 0.0366$. The \cref{eq:pnp_fb} algorithm is ran with $\gamma = 0.2$.

\section{MC sampling and Reynolds averaging}

We compare in Table~\ref{tab:MC_compare_rebuttal} the performance of the equivariant architecture when training with the proposed Monte-Carlo (MC) scheme vs the true averaging. It shows no difference in final performance while the MC strategy decrease the computational complexity by a factor 4.

\begin{table}[t]
\centering
\footnotesize
\begin{tabular}{@{\hskip 0pt}l @{\hskip 5pt} l @{\hskip 5pt} c  @{\hskip 10pt}c@{\hskip 0pt}}
\hline
Architecture & Dataset & Monte-Carlo Sample & Reynolds Average \\ 
\hline
DnCNN & BSD10  & $30.698 \pm 1.645$ 
& $30.684 \pm 1.645$ 
\\ 
DRUNet & fastMRI & $30.678 \pm 0.740$ 
& $30.646 \pm 0.752$
\\ 
LipDnCNN & Set3C & $32.705 \pm 0.868$
& $32.706 \pm 0.868$ 
\\ \hline
\end{tabular}
\vspace{-0.7em}
\caption{Performance of algorithms from Fig.~4 when relying on Monte-Carlo estimates and averaged equivariant architectures.}
\vspace{-1em}
\label{tab:MC_compare_rebuttal}
\end{table}

\section{Equivariant algorithms}

The equivariant counterpart of \eqref{eq:pnp_fb} is
\begin{equation}
\tag{eq. PnP}
\begin{aligned}
&\text{Sample }\,g_k \sim \mathcal{G} \\
&\text{Set }\,\widetilde{\operatorname{D}}_{\G,k}(x) = T_{g_k}^{-1}\operatorname{D}(T_{g_k} x)   \\
&x_{k+1} = \widetilde{\operatorname{D}}_{\G,k}\left(x_k-\gamma A^{\top} (A x_k-y)\right).
\end{aligned}
\end{equation}

\noindent The equivariant counterpart of \eqref{eq:red_gd} is
\begin{equation}
\tag{eq. RED}
\begin{aligned}
&\text{Sample }\,g_k \sim \mathcal{G} \\
&\text{Set }\,\widetilde{\operatorname{D}}_{\G,k}(x)  = T_{g_k}^{-1}\operatorname{D}(T_{g_k} x)   \\
&x_{k+1} = x_k - \gamma A^{\top} (A x_k-y) \\ 
& \qquad \qquad \qquad \qquad - \gamma \lambda (x_k-\widetilde{\operatorname{D}}_{\G,k}(x_k)).
\end{aligned}
\end{equation}

\noindent The equivariant counterpart of \eqref{eq:ula} is
\begin{equation}
\tag{eq. ULA}
\begin{aligned}
&\text{Sample }\,g_k \sim \mathcal{G} \\
&\text{Set }\,\widetilde{\operatorname{D}}_{\G,k}(x)  = T_{g_k}^{-1}\operatorname{D}(T_{g_k} x)   \\
&x_{k+1} = x_k - \gamma A^{\top} (A x_k-y) \\ 
& \qquad \qquad - \gamma \lambda (x_k-\widetilde{\operatorname{D}}_{\G,k}(x_k)) + \sqrt{2\gamma} \epsilon_k.
\end{aligned}
\end{equation}

\end{document}

%% file: preamble.tex
\usepackage[dvipsnames]{xcolor}
